\documentclass{article}

\usepackage{amsmath}
\usepackage{amssymb}
\usepackage{amsthm}
\usepackage{cspsymb}
\usepackage{zed}

\usepackage{lineno,hyperref}
\modulolinenumbers[5]

\usepackage{listings}

\usepackage{multirow}

\usepackage{tikz}
\usetikzlibrary{shapes,shapes.geometric,positioning,arrows,automata, decorations.pathreplacing}

\tikzset{sqstate/.style={draw,regular polygon,regular polygon sides=4,inner sep=0cm}}
\tikzset{rectangle/.style={draw,shape=rectangle,inner sep=0cm}}

\makeatletter
\newtheorem*{rep@thm}{\rep@title}
\newcommand{\newreptheorem}[2]{%
\newenvironment{rep#1}[1]{%
 \def\rep@title{#2 \ref{##1}}%
 \begin{rep@thm}}%
 {\end{rep@thm}}}
\makeatother

\providecommand{\keywords}[1]{\textbf{\textit{Keywords---}} #1}

\theoremstyle{plain}
\newtheorem{thm}{Theorem}
\newreptheorem{thm}{Theorem}
\newtheorem{lem}{Lemma}
\newreptheorem{lem}{Lemma}
\newtheorem*{cor}{Corollary}

\theoremstyle{definition}
\newtheorem{defn}{Definition}

\newtheorem{rexamplex}{Running Example}

\newenvironment{rexample}
{\pushQED{\qed}\rexamplex}
{\popQED\endrexamplex}

\newtheorem{rexamplexx}{Running Example}
\newenvironment{crexample}[1]
{\pushQED{\qed}\renewcommand{\therexamplexx}{#1}\rexamplexx}
{\popQED\endrexamplexx}

\newtheorem{example}[rexamplex]{Example}

\theoremstyle{remark}

\numberwithin{subcase}{case}
\numberwithin{subsubcase}{subcase}
\makeatletter
\@addtoreset{case}{thm}
\makeatother

\newcommand{\cspm}[1]{\texttt{#1}}
\newcommand{\qedb}[1]{{\hfill $\blacksquare$ #1}}

\newcommand{\Seq}[1]{\langle #1 \rangle}

\newcommand{\req}{request}
\newcommand{\ungReq}{ungranted\_request}
\newcommand{\ung}{ungrantedness}
\newcommand{\inVoc}{in\_vocabulary}

\newcommand{\ur}{\mathrel{\rightarrow\hspace{-4pt}\bullet}}

\newcommand{\clr}{client\_request}
\newcommand{\svr}{server\_request}
\newcommand{\clp}{client\_response}
\newcommand{\svp}{server\_response}

\newcommand{\bigo}{\mathcal{O}}

\begin{document}

\title{A Pattern-based deadlock-freedom analysis strategy for concurrent systems\thanks{The EU Framework 7 Integrated Project COMPASS~(Grant Agreement 287829) financed most of the work presented here. This work is partially funded by INES, grants CNPq/465614/2014-0 and FACEPE/APQ/0388-1.03/14. No new	primary data was created as part of the study reported here.}}

\author{Pedro Antonino \hspace{3cm} Augusto Sampaio \\
	        \small{Centro de Inform\'{a}tica, Universidade Federal de Pernambuco, Recife, PE, Brazil} \\
	        \small{\{prga2@cin.ufpe.br, acas@cin.ufpe.br\}} 
	        \vspace{.3cm}\\
	        Jim Woodcock  \\
	        \small{Department of Computer Science, University of York, York, UK} \\
        	\small{jim.woodcock@york.ac.uk}
        }
    
\date{January 31, 2019}

	

\maketitle

\begin{abstract} Local analysis has long been recognised as an effective tool to combat the state-space explosion problem. In this work, we propose a method that systematises the use of local analysis in the verification of deadlock freedom for concurrent and distributed systems. It combines a strategy for system decomposition with the verification of the decomposed subsystems via adherence to behavioural patterns. At the core of our work, we have a number of CSP refinement expressions that allows the user of our method to automatically verify all the behavioural restrictions that we impose. We also propose a prototype tool to support our method. Finally, we demonstrate the practical impact our method can have by analysing how it fares when applied to some examples.   
\end{abstract}

\keywords{ CSP; model checking; refinement; local analysis; behavioural patterns; system decomposition; deadlock freedom}

\section{Introduction}

A deadlock is a long-standing, common pathology of concurrent systems~\cite{Coffman71,Dijkstra02}. It occurs when the system reaches a state where all its components are stuck. The importance of deadlock analysis is attested by the fact that deadlock freedom is often considered to be the first step towards correctness for distributed and concurrent systems. Moreover, safety properties can be reduced to deadlock checking~\cite{Godefroid93}. As with many properties of concurrent systems, deadlock verification can be severely affected by the state space explosion problem~\cite{Baier08}.

One common way to cope with the state space explosion problem is to use local analysis~\cite{Roscoe87,Brookes91,Martin96,Attie05,Gruner10,Ramos11,Antonino14a,Antonino14b,Oliveira2016,Filho16,Antonino16a,Antonino16b,Antonino17a,Antonino17b,Otoni17,Attie18}. Instead of checking the entire state space of the concurrent system, the analysis of small combinations of components is carried out to determine whether a system is deadlock free. In fact, for some complex systems, a method using local analysis might be the only practicable option. Local analysis methods are usually incomplete in the sense that they either guarantee deadlock freedom or are inconclusive. The latter means they can neither show that the system deadlocks nor prove deadlock freedom. Traditional local analysis techniques consist of either fully automatic \emph{a posteriori} verification methods, or guidelines to the design of a system that, if followed, guarantee deadlock freedom by construction. The former techniques do not provide any guidance on how system designers can avoid deadlocks, whereas the latter ones do not provide automatic ways of checking that the guidelines were correctly followed.

We propose a method that provides both guidelines to construct deadlock-free systems and a procedure for automatically checking that the guidelines have been correctly followed. This method embodies a notion of decomposition that can be used to prove deadlock freedom for systems with an acyclic communication topology. Moreover, it relies on three behavioural patterns to deal with cyclic-topology systems. These behavioural patterns restrict both the behaviour of components and the structure of the system. Both the decomposition and behavioural patterns rely on local behavioural analysis. The decomposition relies on the analysis of pairs of components of the system, whereas the behavioural patterns constrain the behaviour of individual components. So, this method is not hindered by the state space explosion problem. Nevertheless, its efficiency comes at the price of incompleteness: deadlock freedom can only be proved for systems that fall into our decomposition/pattern-adherence method.

This proposed method is based on prior works that explored local analysis for the verification of deadlock freedom for concurrent and distributed systems. In fact, both the decomposition strategy and two out of the three patterns presented have been proposed decades ago~\cite{Roscoe87,Martin96}. Nevertheless, we introduce a CSP formalisation based on refinement expressions that can be automatically checked by a refinement checker. In prior works, the decompositions and pattern adherence are characterised in terms of semantic properties that a system must have~\cite{Roscoe87,Martin96}. This characterisation forces the user of such methods to understand not only the formalism but also the subtleties of its semantic models. On the other hand, our characterisation based on refinement expressions, together with design guidelines and tool support, gives a more practical support for the system designer. More importantly, prior works do not suggest an automatic way to test whether a system has a given semantic property, whereas our refinement expressions can be automatically checked by a refinement checker like FDR~\cite{fdr3}. Finally, we conduct some experiments to measure the efficiency gains on the analysis of some practical examples.  

This work is a significant extension of two previous works~\cite{Antonino14a,Antonino14b}. The new contributions of this paper are as follows.

\begin{itemize}
	\item We present formal proofs that our adherence to communication patterns guarantees deadlock freedom.
	
	\item We propose a method that systematises the application of system decomposition and pattern-adherence checking strategy to ensure deadlock freedom for a system. This systematisation should guide the user in applying our method in practice.
	
	\item We analyse the computational complexity of the proposed method, illustrate its application to three systems, and compare the method with three other approaches to deadlock analysis. 
	
	\item We implemented a prototype tool that supports the proposed systematisation, saving a lot of manual effort that would otherwise be required in applying our method.
\end{itemize}

This paper is organised as follows. Section \ref{sec:background} introduces the CSP notation, some of its semantic models and a theory of networks of processes, based on CSP, that we use to represent and reason about concurrent systems. In Section \ref{sec:decomposition}, we present our decomposition strategy and how it can ensure deadlock freedom for acyclic-topology systems. Section \ref{sec:patterns} presents the formalisation of three behavioural patterns that prevent deadlocks. In Section \ref{sec:method}, we present our method and a tool to support it. This section also presents the results of some experiments we conduct to assess the efficiency of our method when compared to traditional \emph{a posteriori} verification techniques. Section~\ref{sec:related} introduces a series of related works and how they relate to ours. Finally, in Section \ref{sec:conclusion}, we present our concluding remarks.

\section{Background}
\label{sec:background}

We present a brief introduction to CSP, including the main operators and semantic models used in this work. Then we introduce a notion of live-network model, which is basically a sequence of components that obey some relevant properties. Two CSP models, used as running examples, are also presented.   

\subsection{CSP}

Communicating Sequential Processes (CSP)~\cite{Hoare85,Roscoe98,Roscoe10} is a notation used to model concurrent systems where processes interact by exchanging messages. In this notation, sequential processes can be combined using high-level parallel operators to create complex concurrent processes. The CSP notation used here is the machine-readable version called $CSP_M$, which is the standard version for encoding CSP processes by the FDR tool~\cite{fdr3}. In the following, we informally introduce some operators of this language using two CSP systems that serve as running examples throughout this paper. Our first example introduces a ring-buffer system.

\begin{rexample} \label{ex:buffer} A ring buffer with $\textit{NCELLS}$ storage cells is a system that stores data in a \emph{first-in-first-out} fashion and where its storage cells are written to in a cyclic way. Cells are organised as if they were part of a ring, and once some piece of data is written to a cell, the next piece of data will be written to the next cell on this ring, provided the next cell is available. Our system can store up to $N = \textit{NCELLS} + 1$ pieces of data because it has an extra cache storage space. $NCELLS$ (and $N$) is a global constant that serves as a parameter for our model; by changing $NCELLS$, we can create an arbitrary-sized system with $NCELLS > 0$ many cells. We use a central controller (described by process \verb|Controller(cache,size,top,bot)| whose parameters are initially $0$) to manage input and output requests to the buffer. This process has four parameters: \verb|cache| holds the next element to be output, \verb|size| keeps track of how many cells are full, \verb|top| and \verb|bot| keep track of which cell is the top (i.e., beginning) and the bottom (i.e., end) of the buffer, respectively. The parameters of a process represent its internal state. 
	
\begin{lstlisting}
Controller(cache,size,top,bot) =
	size < N & Input(cache,size,top,bot) 
	[] 
	size > 0 & Output(cache,size,top,bot)
\end{lstlisting}

The process \verb|c & P| behaves like \verb|P| if the condition \verb|c| is true and like \verb|STOP| if \verb|c| evaluates to false, where \verb|STOP| is the atomic process that does nothing and deadlocks. The process \cspm{P [] Q} represents the external choice of \verb|P| and \verb|Q|, that is, the behaviours of \verb|P| and \verb|Q| are initially offered and then either \verb|P| or \verb|Q| is chosen. We point out that an external choice between \verb|P| and \verb|STOP| behaves just as \verb|P|, i.e., \cspm{P [] STOP} = \verb|P|. So, process \verb|Controller| offers the choice of behaving as \verb|Input| if $size < N$ and as \verb|Output| if $size > 0$.     

If the buffer is not full (i.e., $size < N$), the controller can receive and store some data as described by process \verb|Input|. 
\begin{lstlisting}
Input(cache,size,top,bot) =
	input?x -> 
		(size == 0 & Controller(x,1,top,bot)
		[]
		size > 0 & write.top!x -> 
			Controller(cache,size+1,(top+1)%NCELLS,bot))
\end{lstlisting}
The prefixed process \cspm{a -> P} initially offers the event \verb|a| and after this event is performed it behaves as \verb|P|. $CSP_M$ also proposes the notion of a \emph{channel} that transmits data. A channel \verb|ch| is associated to the type of data, say values in the set $datatype$, they transmit. So, a channel gives rise to a number of events each of which denotes the transmission of a different piece of data, that is, event \verb|ch.x| where $x \in datatype$ denotes the transmission of value $x$. A channel \verb|ch| can output \verb|ch!x| and input \verb|ch?x| values. Outputting \verb|ch!x| simply creates event \verb|ch.x| based on the value $x$, whereas the input operation \verb|ch?x| binds the values of \verb|ch|'s datatype to $x$ (intuitively, this means that \verb|ch| can receive/input any value associated with this channel). So, \verb|ch!x -> P| behaves as a simple prefix, whereas \verb|ch?x -> P| behaves like an external choice: each possible value $v$ for $x$ gives rise to a new branch for which $x = v$. Note that channels and their operations are just syntactic sugar over events. For instance, process \verb|Input| initially inputs some value $v$ on channel \verb|input|, and then it proceeds execution with $x = v$. The expression \verb|(top+1)%NCELLS| stands for the increment of \verb|top| modulo \verb|NCELLS|.

If the buffer is not empty, the controller can output and update its state as described by process \verb|Output|.
\begin{lstlisting}
Output(cache,size,top,bot) =
	output!cache -> 
		(size > 1 & (read.bot?x -> 
			Controller(x,size-1,top,(bot+1)%NCELLS))
		[]
		size == 1 & Controller(cache,0,top,bot))
\end{lstlisting}

The process \verb|Cell(id,0)| describes the individual cells that build up the buffer's storage space. It holds some value which can be read (using channel \verb|read.id|) and updated (using channel \verb|write.id|).
\begin{lstlisting}
Cell(id,val) = 
	read.id!val -> Cell(id,val) 
	[] 
	write.id?x -> Cell(id,x)
\end{lstlisting}

Our final system, given by process \verb|RingBufferBehaviour|,  runs our controller process in parallel with $NCELLS$ storage cell processes using the indexed version of CSP's alphabetised-parallel operator. 

\begin{center}
	\verb"RingBufferBehaviour = || i : {0..NCELL} @ A(i) [P(i)]"
\end{center}
where
\begin{itemize}
	\item \verb|P(0) = Controller(0,0,0,0)|
	\item \verb"A(0) = {|read, write, input, output|}",
	\item \verb|P(i) = Cell(i,0)| for $i \in \{1 \ldots NCELL\}$
	\item \verb"A(i) = {|read.i, write.i|}" for $i \in \{1 \ldots NCELL\}$
	\begin{itemize}
		\item In $CSP_M$, the extension operator $\{|e_1,\ldots,e_n|\}$ gives the events that extend the elements $e_i$. For instance, in this example, we have $\{|\cspm{read}|\}$ gives $\{\cspm{read.i.v} \mid i \in \{0\ldots N-1\} \land v \in \{0,1\} |\}$, assuming cells store binary values $v$.  
	\end{itemize}
\end{itemize}

The parallel process \cspm{P [X||Y] Q} allows $P$ and $Q$ to freely perform events not in the set of events $X \inter Y$, but to perform an event in $X \inter Y$, $P$ and $Q$ must synchronise on it. Additionally, $P$ ($Q$) is only allowed to perform events in $X$ ($Y$). $X$ is called the alphabet of $P$. This parallel operator also has an indexed version \cspm{ || e : S @ [A(e)] P(e)}, where \cspm{A(e)} gives an alphabet and \cspm{P(e)} gives a CSP process. For this indexed version, all processes are put in parallel using their corresponding alphabet. Similar to the binary version of this operator, shared events require synchronisation by all processes having the event on their alphabet and the non-shared events can be performed freely by a process. \verb|RingBufferBehaviour| ensures that components synchronise on shared events, namely, read and write events only occur when the controller and the cells cooperate. We formally define the parallel composition for this system when we later introduce our network model.
\end{rexample}

The second running example that we use describes the well-known asymmetric solution to the dining philosophers problem.

\begin{rexample} \label{ex:philosophers} In the dining philosophers setting, $N$ philosophers are trying to eat on a shared round table; $N$ is a constant that also serves as a parameter for our example/model. To do so, each of them must acquire a pair of forks: one on its left-hand side and another on its right-hand side. Philosophers share their right-hand fork with their right neighbour and their left-hand one with the left neighbour. If all philosophers acquire their forks in the same order, they might run into the following deadlock. Say that all philosophers acquire first their left-hand fork and then their right-hand one, then they might reach a state where all of them have acquired their left-hand fork and are waiting for their right-hand one to be released. A well-known solution to avoid this deadlock is to have an asymmetric philosopher that acquires forks in the opposite order.
	
We describe the behaviour of philosophers that acquire and release first their left-hand fork and then their right-hand one by process \verb|Phil(id)|. On the other hand, asymmetric philosophers are described by \verb|APhil(id)|. Event \verb|pickup.i.j| (\verb|putdown.i.j|) is used by philosopher $i$ to acquire (release) fork $j$. Functions \verb|next(i)| and \verb|prev(i)| yield $(i + 1)\%N$ and $(i - 1)\%N$, respectively.
\begin{lstlisting}
Phil(id) = 
	sit.id -> pickup.id.id -> pickup.id!next(id) -> eat.id -> 
		putdown.id.id -> putdown.id!next(id) -> getup.id -> 
			Phil(id)

APhil(id) = 
	sit.id -> pickup.id!next(id) -> pickup.id.id -> 
		eat.id -> putdown.id!next(id) -> putdown.id.id -> 
			getup.id -> APhil(id)
\end{lstlisting}

A fork can be acquired by a philosopher which later releases it as described by process \verb|Fork(id)|.
\begin{lstlisting}
Fork(id) = 
	[] i : {id,prev(id)} @ pickup.i.id -> putdown.i.id -> Fork(id)
\end{lstlisting}
The process \cspm{[] x : S @ P(x)} is the indexed version of the external choice operator. For $\verb|S| = \{v_1,\ldots,v_{|S|}\}$ where $|S|$ gives the size of set $S$, this process is \cspm{P($v_1$) [] \ldots\ [] P($v_{|S|}$)}. 

The system implementing the asymmetric solution, given by \hbox{\verb|APhilsBehaviour|}, runs in parallel $N$ forks, $N-1$ philosophers and an asymmetric philosopher. It relies on the indexed version of CSP's alphabetised-parallel operator to ensure processes synchronise on shared events. 


\begin{center}
	\verb"APhilsBehaviour = || i : {0..2N-1} @ A(i) [P(i)]"
\end{center}
where
\begin{itemize}
	\item \verb|P(i) = Phil(i)| for $i \in \{0,\ldots,N-2\}$
	\item \verb|A(i) = |$AlphaPhil(i)$ for $i \in \{0,\ldots,N-2\}$
	\begin{itemize}
		\item $AlphaPhil(i) = \begin{array}[t]{l} \{\cspm{sit.i},\cspm{pickup.i.i},\cspm{pickup.i.next(i)},\cspm{eat.i},\\
		\quad \cspm{putdown.i.i},\cspm{putdown.i.next(i)},\cspm{getup.i}\} \end{array}$
	\end{itemize}
	\item \verb|P(N-1) = APhil(N-1)|
	\item \verb|A(N-1) = |$AlphaPhil(N-1)$
	\item \verb|P(i) = Fork(i)| for $i \in \{N,\ldots,2N-1\}$
	\item \verb|A(i) = |$\begin{aligned}[t]&(\{\cspm{pickup.i.i},\cspm{pickup.prev(i).i},\\ & \qquad\cspm{putdown.i.i},\cspm{putdown.prev(i).i}\},\cspm{Fork(i)})
	\end{aligned}$
\end{itemize}

\end{rexample}

\subsection{Denotational semantics}

In order to reason about processes, CSP embodies a collection of mathematical
models. In this work, we use the \textit{stable failures}
model, and the less conventional \textit{stable revivals} model.

In the stable-failures model, a process is represented by
a pair $(F,T)$ containing its stable failures and its finite traces, respectively. The traces of a process are represented by a set of all the finite sequences of visible events that this process can perform; this
set is given by $\traces(P)$. The stable failures of a process are represented
by a set of pairs $(s,X)$, where $s$ is a trace and $X$ is a set of
events that the process can refuse to do after performing the trace $s$. At the
state where the process can refuse events in $X$, the process must not be able
to perform an internal action, otherwise this state would be unstable and would
not be taken into account in this model. The function $\failures(P)$ gives the
set of stable failures of process $P$. Hence, the representation of process $P$
in this model is given by the pair $(\failures(P),traces(P))$.

Before introducing how to systematically calculate the traces and failures of a process, we introduce a few more constructs of the $CSP_M$ notation. Similarly to \cspm{STOP}, \cspm{SKIP} is the atomic process that does nothing and terminates successfully. Another useful atomic process, mainly from a theoretical perspective, is \cspm{div}, which is the diverging process. $\Sigma$ is the universal set of visible events; the invisible event $\tau$ and the termination signal $\tick$ are not members of this set. The internal (non-deterministic) choice process \cspm{P |\~{}| Q} offers either \verb|P| or \verb|Q| non-deterministically. The process \cspm{P ; Q} behaves initially as process \cspm{P} and, once \cspm{P} successfully terminates, it behaves as process \cspm{Q}.

The renaming process \cspm{P [[R]]}, where \cspm{R} is a set of pairs \cspm{a <- b}, offers a deterministic choice of events in \cspm{S} whenever \cspm{P} offers \cspm{a}, where  $\cspm{S} = \{ \cspm{b} \mid \cspm{(a <- b)} \in \cspm{R}\}$. The hidden process \verb"P \ S" offers the events not in $S$ whenever \cspm{P} offers them. On the other hand, \verb"P \ S" can perform a $\tau$, the silent event, whenever \cspm{P} can perform an event in \verb"S". The interrupt process \cspm{P} \verb|/\| \cspm{Q} behaves like \cspm{P} and at any point it can be interrupted in which case it behaves as \cspm{Q}.

The $CSP_M$ notation does not provide an explicit operator for recursion, but
it allows one to use the name of the process in its definition. For instance, \cspm{P = a -> P}
performs \cspm{a}, and then recurses, behaving as \cspm{P}. Even though a
formal construct is not available for recursion, we can define it as an
equation where the right-hand side is a process context depending on the
definition of the process itself, e.g. $X = F(X)$. For the process \cspm{P} given above, we can define it as \cspm{P = F(P)}, where $F(X) =$\cspm{ a -> }$X$.
For the purpose of giving the semantics of a recursive process, we use this style of
definition.


\begin{table}[!b]
	\setlength{\tabcolsep}{2pt}
	\def\arraystretch{1.3}
	\resizebox{\textwidth}{!}{%
		\begin{tabular}{ll}
		$\traces(\cspm{STOP})$ & $= \{\emptyseq\}$ \\
		$\traces(\cspm{SKIP})$   & $= \{\emptyseq, \seq{\tick} \}$ \\
		$\traces(\cspm{div})$     & $= \{\emptyseq\}$ \\
		$\traces(a\cspm{ -> }P)$ & $= \{\emptyseq\} \cup \{\seq{a}\cat s~|~ s \in
		\traces(P)\}$ \\
		$\traces(P\cspm{ ; }Q) $& $= (\traces(P) \cap \Sigma^{*}) \cup \{s\cat
		t~|~s\cat\seq{\tick} \in \traces(P) \land t \in \traces(Q) \}$ \\
		$\traces(P\cspm{ [] }Q) $& $= \traces(P) \cup \traces(Q)$ \\
		$\traces(P\cspm{ |\~{}| }Q) $& $= \traces(P) \cup \traces(Q)$ \\
		$\traces(P \cspm{ [|}X\cspm{|] } Q) $& $=  \bigcup \{ s \parallel[X] t ~|~
		s \in \traces(P) \land t\in\traces(Q) \}$  \\
		$\traces(P\cspm{ \textbackslash~}X) $& $=  \{ s\hide X~|~ s \in
		\traces(P)\}$
		\\
		$\traces(P\cspm{ [[R]]})$ & $= \{ t~|~ \exists s \in \traces(P) \spot
		s~R^{*}~t\}$\\
		$\traces(P~/\backslash~Q)$ & $= traces(P) \cup \{ s \cat t~|~ s \in \traces(P) \inter \Sigma^* \land t \in \traces(Q)\}$
		\end{tabular}}
		~\\where \begin{itemize} 
			 \item $s \cat t$ represents the concatenation of sequences $s$ and $t$.
			\item $s \parallel[X] t$ gives all the traces $w$ that are interleaving of $s$ and $t$ such that $w \restrict X = s \restrict X = t \restrict X$.
			\item $s \restrict X$ gives the trace resulting from removing events not in $X$ from $s$. 
			\item $s \hide X$ gives the trace resulting from removing events in $X$ from $s$.
			\item $\Seq{s_0,\ldots,s_n} R^* \Seq{t_0,\ldots,t_m}$ holds iff $n = m$ and $ \forall i \in \{0\ldots n\} \bullet s_iRt_i$.
		\end{itemize}
	\caption{Semantic clauses for the traces model}
	\label{tab:semanticClausesTraces}
\end{table}

\begin{table}[!h]
	\setlength{\tabcolsep}{2pt}
	\def\arraystretch{1.3}
	\resizebox{\textwidth}{!}{%
		\begin{tabular}{ll}
		$ \failures(\cspm{STOP})$ & $= \{(\emptyseq, X)~|~X\subseteq \Sigma \union \{\tick\}\}$ \\
		$ \failures(\cspm{SKIP})$ & $= \{(\emptyseq, X)~|~X\subseteq \Sigma\} \cup
		\{(\seq{\tick},X)~|~X\subseteq\Sigma \union \{\tick\}\}$ \\
		$ \failures(\cspm{div})$ & $= \emptyset$ \\
		$ \failures(a \cspm{ -> } P)$ & $= \{(\emptyseq, X)~|~a\nin X \} \cup  
		\{(\seq{a}\cat s, X)~|~(s,X) \in \failures(P) \}$ \\
		$ \failures(P \cspm{ ; } Q)$ & $= \begin{aligned}[t] 	
		&\{(s,X)~|~s\in\Sigma^{*} \land (s, X\cup\set{\tick}) \in \failures(P)\} \cup \\
		 &\{(s\cat t, X)~|~s\cat\seq{\tick}\in\traces(P) \land (t,X) \in \failures(Q) \} \end{aligned}$ \\
		$ \failures(P \cspm{ [] } Q)$ & $=  \begin{aligned}[t] &
		\{(\emptyseq,X)~|~ (\emptyseq,X) \in \failures(P) \cap \failures(Q)\} \cup \\
		& \{(t,X)~|~ (t,X) \in \failures(P) \cup \failures(Q) \land t \neq \emptyseq \}\cup \\
		& \{(\emptyseq,X) ~|~ X \subseteq\Sigma \land \seq{\tick} \in \traces(P)\cup \traces(Q) \}       
		\end{aligned} $\\
		$ \failures(P \cspm{ |\~{}| } Q)$ & $= \failures(P) \cup \failures(Q)$ \\
		$ \failures(P \cspm{ [|}X\cspm{|] } Q)$ & $=  \bigcup  \{(s \parallel[X] t,
		Y \cup Z)~|~ \begin{aligned}[t] &Y\backprime (X \cup\set{\tick}) = Z \backprime (X\cup\set{\tick}) \land \\
		& (s,Y) \in \failures(P) \land (t,Z) \in \failures(Q)  \} \end{aligned}$ \\
		$ \failures(P \cspm{ \textbackslash~} X)$ & $= \{(t \hide X,Y) ~|~ (t,Y
		\cup X) \in failures(P)\}$ \\
		$ \failures(P \cspm{ [[}R\cspm{]]})$ & $= \{(t,X)~|~ (\exists t' ~|~ (t',
		t) \in R^{*} \land (t',R^{-1}(X)) \in \failures(P)\}$\\
		$\failures(P~/\backslash~Q)$ & $= \begin{aligned}[t]
		&\{ (s,X) ~|~ (s,X) \in \failures(P) \land s \in \Sigma^* \land (\emptyseq,X) \in \failures(Q)\} \\ 
		&\cup \{ (s,X) ~|~ s\cat\seq{\tick} \in \traces(P) \land \tick \nin X\} \\
		&\cup \{ (s\cat\seq{\tick},X) ~|~ s\cat\seq{\tick} \in \traces(P)\} \\
		& \cup \{ (s \cat t,X) ~|~ s \in \traces(P) \inter \Sigma^* \land (t,X) \in \failures(Q)\} \end{aligned}$
		\end{tabular}}
	\caption{Semantic clauses for the failures model}
	\label{tab:semanticClausesFailures}
\end{table}

The functions $\traces(P)$ and $\failures(P)$ are calculated inductively based
on the constructs of the CSP language. The clauses for calculating the $\traces$ are
presented in Table~\ref{tab:semanticClausesTraces}, whereas the clauses
for calculating the $\failures$ are depicted in
Table~\ref{tab:semanticClausesFailures}.
The semantics of a recursive process can be calculated, using the presented clauses, thanks to the following equivalence. For a recursive process $P = F(P)$, $P \equiv \Intchoice \{F^n(\cspm{div}) \mid n \in \Nats\}$, where $\Intchoice S$ is the distributed application of the operator \verb"|~|" to the processes in $S$. This equivalence also holds for the \emph{stable revivals} model, presented later.

We illustrate the calculation of these behaviours using our ring-buffer system.

\begin{crexample}{\ref{ex:buffer}} We illustrate the traces and stable-failures for the processes \verb|Controller(0,0,0,0)| and \verb|Cell(0,0)|. For the following $\failures$ sets, $(tr,S)$ is a shorthand for all pairs $(tr,X)$ such that $X \subseteq S$; this makes our examples more compact. 
	\begin{itemize}
		\item $\traces(\cspm{Controller(0,0,0,0)}) = \\ \qquad \begin{array}[t]{l}
		\{ \Seq{},\Seq{input.0},\Seq{input.1},\Seq{input.2}, \Seq{input.0,input.0}, \Seq{input.0,input.1}, \\ \ \Seq{input.0,input.2}, \Seq{input.0,output.0}, \Seq{input.1,output.1}, \ldots \}\end{array}$
		\item $\failures(\cspm{Controller(0,0,0,0)}) = \\ \qquad \begin{array}[t]{l}
		\{ (\Seq{},\Sigma - \{input.0,input.1,input.2\}),
		\\ \ (\Seq{input.0},\Sigma - \{input.0,input.1,input.2,output.0\}), 
		\\ \ (\Seq{input.1},\Sigma - \{input.0,input.1,input.2,output.1\}), 
		\\ \ (\Seq{input.2},\Sigma - \{input.0,input.1,input.2,output.2\})), \ldots \}\end{array}$
		\item $\traces(\cspm{Cell(0,0)}) = \\ \qquad \begin{array}[t]{l}
		\{ \Seq{},\Seq{read.0},\Seq{write.0},\Seq{write.1},\Seq{write.2},\Seq{read.0,write.0}, \\ \ \Seq{read.0,write.1}, \Seq{read.0,write.2},\Seq{read.0,write.0,read.0},\ldots \}\end{array}$
		\item $\failures(\cspm{Cell(0,0)}) = \\ \qquad \begin{array}[t]{l}
		\{ (\Seq{},\Sigma - \{read.0,write.0,write.1,write.2\}),
		\\ \ (\Seq{read.0},\Sigma - \{read.0,write.0,write.1,write.2\}),
		\\ \ (\Seq{write.0},\Sigma - \{read.0,write.0,write.1,write.2\}),
		\\ \ (\Seq{write.1},\Sigma - \{read.1,write.0,write.1,write.2\})),
		\\ \ (\Seq{write.2},\Sigma - \{read.2,write.0,write.1,write.2\}), \ldots \}\end{array}$
	\end{itemize}
\end{crexample}

In this model, the failures for a given trace are subset closed: if $(s,X) \in \failures(P)$ then so is $(s,Y)$ provided $Y \subseteq X$.
So, for some properties, we will be interested only in the maximal failures, considering the subset order, for each trace $s$. $\overline{\failures}(P)$ denotes the set of such maximal failures for process $P$.

\subsubsection*{Stable Revivals Model}

In the stable revivals model, a process is
described by a triple $(T,D,R)$ containing its traces, its deadlocks and
its stable revivals, respectively. The deadlocks of a process are given by the set of
traces after which the process refuses all the events in its alphabet; this set
is given by $\deadlocks(P)$. The stable revivals set is composed of triples $(s,X,a)$ containing a trace, a refusal set, and a revival event, respectively. The refusal set $X$, similarly to the one described in the stable failures model, describes the set of events that can be refused by the process after the trace $s$. The revival event $a$ represents an event that the process
can offer after performing $s$ and refusing $X$. At the state where the revival
is recorded, the process must not be able to perform an internal action,
otherwise this state is unstable, not being taken into account. The function
$\revivals(P)$ gives the set of stable revivals of process $P$. Thus, the
representation of a process $P$ in this model is given by $(traces(P),
deadlocks(P),revivals(P))$. The necessity of this model comes from the fact that some properties that we intend to capture cannot be naturally captured using the notion of refinement over the failures model. Conflict freedom is a concrete example of such properties. Generally, requiring that different refusal behaviours $R_1$ and $R_2$, where $R_1 \subset R_2$, from a process based on whether a particular event is offered cannot be naturally captured using failures refinement; the subset-closed structure of refusal sets gets in the way of specifying such a property.

The functions $\traces(P)$, $\deadlocks(P)$ and $\revivals(P)$ are calculated
inductively based on the constructs of the CSP language. The clauses for calculating the
$\traces$ are presented in Table~\ref{tab:semanticClausesTraces}. In the same
way, the clauses for calculating $\deadlocks$ and $\revivals$ are depicted in
Table~\ref{tab:semanticClausesDeadlocks} and
Table~\ref{tab:semanticClausesRevivals}, respectively.

\begin{table}[!t]
		\setlength{\tabcolsep}{2pt}
		\def\arraystretch{1.3}
		\resizebox{\textwidth}{!}{%
		\begin{tabular}{ll}
		$ \deadlocks(\cspm{STOP}) $ & $ = \{\emptyseq\}$ \\
		$ \deadlocks(\cspm{SKIP}) $ & $ =  \emptyset$ \\
		$ \deadlocks(\cspm{div}) $ & $ =  \emptyset $ \\
		$ \deadlocks(a \cspm{ -> } P) $ & $ =  \{\seq{a}\cat s~|~ s \in
		\deadlocks(P) \} $ \\
		$ \deadlocks(P \cspm{ ; } Q) $ & $ =  
		\begin{aligned}[t]
		& \{s~|~ s \in \deadlocks(P)\}\cup \{s \cat t~|~s\cat\seq{\tick}\in\traces(P) \land t \in
		\deadlocks(Q) \} 
		\end{aligned} $ \\
		$ \deadlocks(P \cspm{ [] } Q) $ & $ =   
		\begin{aligned}[t] 
		& ((\deadlocks(P) \union \deadlocks(Q)) \inter \\ & \quad \{ s~|~s\neq \emptyseq\})\union (\deadlocks(P) \inter \deadlocks(Q))
		\end{aligned} $ \\
		$ \deadlocks(P \cspm{ |\~{}| } Q) $ & $ =  \deadlocks(P) \cup \deadlocks(Q)
		$ \\
		$ \deadlocks(P \cspm{ [|}X\cspm{|] } Q) $ & $ =  
		\begin{aligned}[t] 
		&\{u~| \exists (s,Y):\failures(P); (t,Z) : \failures(Q) \spot \\ & \qquad Y \setminus (X \union \{\tick\}) = Z \setminus (X \union \{\tick\})
		\land u \in (s \parallel[X] t) \inter \Sigma^* \land \Sigma^{\tick} = Y \union Z \} 
		\end{aligned}$ \\
		$ \deadlocks(P \cspm{ \textbackslash~} X) $ & $ =  \{s \hide X ~|~ s \in
		\deadlocks(P)\} $ \\
		$ \deadlocks(P \cspm{ [[}R \cspm{]]}) $ & $ =  \{ s'~|~ \exists s \spot s R
		s' \land s \in \deadlocks(P)\}$ \\
		where:&\\
		$\quad \failures(P)$ & $ = \{(s,X)~|~X \subseteq \Sigma^{\tick} \land s \in Dead
		\} \union \{(s,X), (s, X \union \{\tick\})~|~(s,X,a) \in Rev \}$ \\
		& $\quad\union \{(s,X)~|~s\cat \seq{\tick} \in Tr \land X \subseteq
		\Sigma \}\union \{(s\cat\seq{\tick},X)~|~s\cat \seq{\tick} \in Tr \land X
		\subseteq \Sigma^{\tick} \}$ \\
		\end{tabular}}
	\caption{Deadlocks semantic clauses}
	\label{tab:semanticClausesDeadlocks}
\end{table}

\begin{table}[!ht]
	\setlength{\tabcolsep}{2pt}
	\def\arraystretch{1.3}
	\resizebox{\textwidth}{!}{%
		\begin{tabular}{ll}
		$ \revivals(\cspm{STOP}) $ & $ =  \emptyset$ \\
		$ \revivals(\cspm{SKIP}) $ & $ =  \emptyset$ \\
		$ \revivals(\cspm{div}) $ & $ =  \emptyset $ \\
		$ \revivals(a \cspm{ -> } P) $ & $ =  \{(\emptyseq, X,a)~|~a\nin X \} \cup  
		\{(\seq{a}\cat s,X,b)~|~(s,X,b) \in \revivals(P) \} $ \\
		$ \revivals(P \cspm{ ; } Q) $ & $ =  
		\begin{aligned}[t]
		& \{(s,X,a)~|~ \land (s,X,a) \in \revivals(P)\}\\
		& \cup \{(s \cat t, X,a)~|~s\cat\seq{\tick}\in\traces(P) \land (t,X,a) \in
		\revivals(Q) \} 
		\end{aligned} $ \\
		$ \revivals(P \cspm{ [] } Q) $ & $ =   
		\begin{aligned}[t] 
		& \{(\emptyseq,X,a)~|~(\emptyseq,X) \in \failures^b(P) \cap \failures^b(Q)\\
		& \quad \land (\emptyseq,X,a) \in \revivals(P) \union \revivals(Q)\}\\
		& \cup \{(s,X,a)~|~ (s,X,a) \in \revivals(P) \cup \revivals(Q) \land s \neq
		\emptyseq\} \end{aligned} $ \\
		$ \revivals(P \cspm{ |\~{}| } Q) $ & $ =  \revivals(P) \cup \revivals(Q) $ \\
		$ \revivals(P \cspm{ [|}X\cspm{|] } Q) $ & $ =  
		\begin{aligned}[t] 
		&\{(u,Y \cup Z,a)~| \exists s,t \spot (s,Y) \in \failures^b(P) \land (t,Z) \in \failures^b(Q)\\
		& \land u \in s \parallel[X] t \land Y \setminus X = Z \setminus X \\
		& \land (( a \in X \land (s,Y,a) \in \revivals(P) \land (t,Z,a) \in \revivals(Q)) \\
		& \quad \lor (a \nin X \land ((s,Y,a) \in \revivals(P) \lor (t,Z,a) \in \revivals(Q))))\} \\
		& \union \{(u,Y \cup Z,a)~| \exists s,t \spot (s,Y,a) \in \revivals(P) \land t \cat \seq{\tick} \in
		\traces(Q) \\
		& \land Z \subseteq X \land a \nin X \land u \in s \parallel[X] t\} \\
		& \union \{(u,Y \cup Z,a)~| \exists s,t \spot (t,Z,a) \in \revivals(Q) \land s \cat \seq{\tick} \in
		\traces(P) \\
		& \land Y \subseteq X \land a \nin X \land u \in s \parallel[X] t\} 
		\end{aligned}$ \\
		$ \revivals(P \cspm{ \textbackslash~} X) $ & $ =  \{(s \hide X,Y,a) ~|~ (s,Y
		\cup X,a) \in \revivals(P)\} $ \\
		$ \revivals(P\cspm{ [[}R\cspm{]]}) $ & $ =  \{(s',X,a')~|~ \exists s,a \spot
		s R s' \land a R a' \land (s,R^{-1}(X),a) \in \revivals(P)\}$ \\
		where:&\\
		$\quad \failures^b(P)$ & $ = \{(s,X)~|~X \subseteq \Sigma \land s \in Dead
		\}\union \{(s,X)~|~(s,X,a) \in Rev \}$
		\end{tabular}}
	\caption{Revivals semantic clauses}
	\label{tab:semanticClausesRevivals}
\end{table}

We illustrate the calculation of deadlocks and revivals using our ring-buffer system.

\begin{crexample}{\ref{ex:buffer}} We illustrate the revivals for the processes  \verb|Cell(0,0)| and \verb|Controller(0,0,0,0)|. Since these two processes do not deadlock, they have empty $deadlocks$ sets. For the following $revivals$ sets, to make our presentation more compact, we use $(tr,S,S')$ as a shorthand denoting all pairs $(tr,X, a)$ such that $X \subseteq S$ and $a \in S'$. 
	\begin{itemize}
		\item $revivals(\cspm{Controller(0,0,0,0)}) = \\ \qquad \begin{array}[t]{l}
		\{ (\Seq{},\Sigma - \{input.0,input.1,input.2\}, \{input.0,input.1,input.2\}),\\ 
		\ (\Seq{input.0},\Sigma - \{input.0,input.1,input.2,output.0\}, \\ 
		\qquad \{input.0,input.1,input.2,output.0\}), \\ 
		\ (\Seq{input.1},\Sigma - \{input.0,input.1,input.2,output.1\},\\
		\qquad  \{input.0,input.1,input.2,output.1\}), \\ 
		\ (\Seq{input.2},\Sigma - \{input.0,input.1,input.2,output.2\}, \\
		\qquad  \{input.0,input.1,input.2,output.2\}), \ldots \}\end{array}$
		
		\item $revivals(\cspm{Cell(0,0)}) = \\ \qquad \begin{array}[t]{l}
		\{ (\Seq{},\Sigma - \{read.0,write.0,write.1,write.2\}, 
		\\ \qquad\{read.0,write.0,write.1,write.2\}),
		\\\ (\Seq{read.0},\Sigma - \{read.0,write.0,write.1,write.2\}, 
		\\ \qquad \{read.0,write.0,write.1,write.2\}),
		\\ \ (\Seq{write.0},\Sigma - \{read.0,write.0,write.1,write.2\}, 
		\\ \qquad \{read.0,write.0,write.1,write.2\}),
		\\ \ (\Seq{write.1},\Sigma - \{read.1,write.0,write.1,write.2\}, 
		\\ \qquad \{read.1,write.0,write.1,write.2\} ),
		\\ \ (\Seq{write.2},\Sigma - \{read.2,write.0,write.1,write.2\}, 
		\\ \qquad \{read.2,write.0,write.1,write.2\}),\ldots \}\end{array}$
	\end{itemize}
\end{crexample}

The CSP framework offers, for each semantic model, a refinement relation between processes. \cspm{[F=} is the refinement relation for the stable failures model. \cspm{P [F= Q} holds if and only if $traces(Q) \subseteq traces(P)$ and $\failures(Q) \subseteq \failures(P)$ hold. This order relation can be seen as depicting that \cspm{Q} is more deterministic than \cspm{P}. \cspm{[V=} is the refinement relation for the stable revivals model. \cspm{P [V= Q} holds if and only if $traces(Q) \subseteq traces(P)$, $revivals(Q) \subseteq revivals(P)$ and $deadlocks(Q) \subseteq deadlocks(P)$ hold. This relation can be seen as depicting a finer more-deterministic order. While \cspm{P [F= Q} establishes that \cspm{Q} is more deterministic than \cspm{P} after each trace, \cspm{P [V= Q} establishes that \cspm{Q} is more deterministic than \cspm{P} for each event offered after each trace (namely, \cspm{Q} must refuse fewer events than \cspm{P} for each offer of an event $a$ after the trace $s$). 

\subsection{Network model}

The concepts presented in this section are essentially a slight reformulation of concepts presented in
\cite{Brookes91,Roscoe87}. A \emph{network} provides a model for a concurrent system in terms of its components. 

\begin{defn} A network is a sequence of components $\Seq{C_1,\ldots,C_n}$, where $C_i = (A_i,P_i)$, $A_i \subseteq \Sigma$ and  $P_i$ is a CSP process.
\end{defn}

In this work, we consider only \emph{live} networks. A network is live if and only if it is busy, non-terminating and triple disjoint. A network is busy if and only if every component is deadlock free, non-terminating if and only if every component does not terminate, and triple disjoint if and only if an event is shared by at most two components.

The \emph{communication topology} (or topology for short) of a network can be analysed through its \emph{communication graph}. It shows how components are connected, where a connection (i.e. edge) between two components represent that they (might) communicate/interact. This graph only depicts the (static) connections between components.

\begin{defn}
	The communication graph of a network is an undirected graph where nodes denote components and there is an (undirected) edge between two nodes if and only if the corresponding components share some event.
\end{defn} 

For example, Figure~\ref{fig:comm_snap_graph} depicts the communication graph for the system implementing the (deadlocking) symmetric version of the dining philosophers problem with 3 philosophers and 3 forks, Figure~\ref{fig:ringbuffer} depicts the communication graph for an instance of our ring-buffer network with 3 storage cells, and Figure~\ref{fig:diningphilosopher} depicts the communication graph for our (asymmetric) dining-philosophers network with 3 philosophers and 3 forks.

Note that this graph can be constructed based on a static analysis of components and their alphabets. The communication topology of a network plays an important role in deadlock analysis as we present later.

The behaviour of a network is given by a composition of the components' behaviours as follows.

\begin{defn} Let $V = \Seq{C_1,\ldots,C_n}$, where $C_i = (A_i,P_i)$, be a network. The behaviour of $V$ is given by the CSP expression: \cspm{|| $i$ : $\{1\ldots n\}$ @ [$A_i$] $P_i$}
	\label{def:behaviour}
\end{defn}

We (re-)define the systems in our running examples using the network model. Note how the behaviour of the following networks coincide with the processes that we have earlier introduced to capture the behaviour of the systems in our running examples. We define our ring buffer system as follows.

\begin{crexample}{\ref{ex:buffer}} Our ring-buffer system is defined by the \cspm{RingBuffer} network. Its behaviour requires processes to synchronise on shared events.
	\begin{center}
		$\cspm{RingBuffer} = \Seq{Controller,Cell(0),\ldots,Cell(N-1)}$
	\end{center}
	\begin{itemize}
		\item $Controller = (\{|\cspm{read},\cspm{write},\cspm{input},\cspm{output}|\},\cspm{Controller(0,0,0,0)})$
		\item $Cell(i) = (\{|\cspm{read.i},\cspm{write.i}|\},\cspm{Cell(i,0)})$
		\begin{itemize}
			\item In $CSP_M$, the extension operator $\{|e_1,\ldots,e_n|\}$ gives the events that extend the elements $e_i$. For instance, in this example, we have $\{|\cspm{read}|\}$ gives $\{\cspm{read.i.v} \mid i \in \{0\ldots N-1\} \land v \in \{0,1\} |\}$, assuming cells store binary values $v$.  
		\end{itemize}
	\end{itemize}
\end{crexample}

The asymmetric solution to the dining philosophers problem is defined as follows.

\begin{crexample}{\ref{ex:philosophers}} We define our asymmetric solution system using network \cspm{APhils}. Note that philosopher $N-1$ behaves asymmetrically. Also, we point out that this network's behaviour requires processes to synchronise on shared events.
	
	\begin{center}
		$\cspm{APhils} = \Seq{Phil(0),\ldots,Phil(N-2),APhil(N-1),Fork(0),\ldots,Fork(N-1)}$
	\end{center}
	\begin{itemize}
		\item $Phil(i) = (AlphaPhil(i),\cspm{Phil(i)})$
		\begin{itemize}
			\item $AlphaPhil(i) = \begin{array}[t]{l} \{\cspm{sit.i},\cspm{pickup.i.i},\cspm{pickup.i.next(i)},\cspm{eat.i},\\
			\quad \cspm{putdown.i.i},\cspm{putdown.i.next(i)},\cspm{getup.i}\} \end{array}$
		\end{itemize}
		\item $APhil(i) = (AlphaPhil(i),\cspm{APhil(i)})$
		\item $Fork(i) = \begin{aligned}[t]&(\{\cspm{pickup.i.i},\cspm{pickup.prev(i).i},\\ & \qquad\cspm{putdown.i.i},\cspm{putdown.prev(i).i}\},\cspm{Fork(i)})
		\end{aligned}$
	\end{itemize}
\end{crexample}

To reason about the behaviour of a network, we define the notion of a \emph{state}. A state presents an instant picture of the behaviour of the network in terms of its components' behaviours.

\begin{defn} Let $V = \Seq{C_1,\ldots,C_n}$ be a network where $C_i = (A_i,P_i)$. A state of the network is a pair $(s,R)$, where $R = (R_1,\ldots,R_n)$, such that:
	\begin{itemize}
		\item $s \in \Sigma^*$
		\item For all $i \in \{1\ldots n\}$, $(s \project A_i,R_i) \in \overline{\failures}(P_i)$.
	\end{itemize}
\end{defn}

A network deadlocks if and only if it can reach a state in which no further action can be taken. 

\begin{defn} Let $V = \Seq{C_1,\ldots,C_n}$ be a network where $C_i = (A_i,P_i)$, and $\sigma= (s,R)$, where $R = (R_1,\ldots,R_n)$, one of $V$'s states.
	\begin{center}
		$Deadlocked(\sigma) \defs Refusals(\sigma) = \Sigma$
	\end{center}
	where $Refusals(\sigma) \defs \Union\{ A_i \inter R_i \mid i \in \{1\ldots n\}\}$
\end{defn}

For live networks, ungranted requests are considered to be the building blocks of deadlocks. An ungranted request denotes a wait-for dependency from a component to another component in a given state. It arises, in a system state, when one component is offering an event which is being refused by its communication partner, so they cannot synchronise on this event. As components in a live network do not deadlock, a deadlock must be formed of a situation in which there exists a mutual wait between components.

\begin{defn} Let $V = \Seq{C_1,\ldots,C_n}$ be a network where $C_i = (A_i,P_i)$, and $\sigma= (s,R)$, where $R = (R_1,\ldots,R_n)$, one of $V$'s states. There is an ungranted request between $i$ and $j$ in state $\sigma$ if the following predicate holds.
	\begin{flalign*}
	&\ungReq(\sigma,i,j) \defs \\& \qquad \req(\sigma,i,j) \land \ung(\sigma,i,j) \land
	\inVoc(\sigma,i,j)&
	\end{flalign*}
	where:
	\begin{itemize}
		\item $\req(\sigma,i,j) \defs (A_i - R_i) \inter
		A_j \neq \emptyset$
		\item $\ung(\sigma,i,j) \defs (A_i - R_i)
		\inter (A_j - R_j) = \emptyset$
		\item $\inVoc(\sigma,i,j) \defs (A_i - R_i)
		\union (A_j - R_j) \subseteq \textit{Voc}$
		\begin{itemize}
			\item $\textit{Voc} =\bigcup_{i,j \in \{1 \ldots n\} \land i \neq j} (A_i \inter A_j)$ gives the vocabulary of the network, namely, the events requiring components to synchronise.
		\end{itemize}
	\end{itemize}
\end{defn}

Given a fixed state $\sigma$, we use $i \ur j$ to denote that there exists an ungranted request from $i$ to $j$ in $\sigma$. For a given fixed state, one can calculate all ungranted requests between components and create what we call a \emph{snapshot graph}. 

\begin{defn} A snapshot graph for system state $\sigma$ is a directed graph where components are nodes and there is an edge from component $i$ to component $j$ if and only if there is an ungranted request in $\sigma$ from $i$ to $j$. 
\end{defn}
	
Unlike communication graphs that depict a static view of the (fixed) topology of the network, these snapshot graphs give instantaneous pictures of the dynamic behaviour of the system. Instead of attempting to capture the overall complexity of components' interactions, a snapshot graph depicts dependencies (i.e., ungranted requests) between components, which are the building blocks for our study of deadlock.

As mentioned, a network deadlocks when all components are blocked in a given state. For a live network, in such a state, all components must be waiting for some other component to advance. This situation implies that there must exist a cycle of ungranted requests between components. To be more concrete, the snapshot graph constructed for a deadlocked system state must exhibit a cycle (of dependencies). The following theorem, asserting these two facts, is our main tool in proving the soundness of our framework. These facts and their proofs can be found in~\cite{Roscoe87,Martin96,Roscoe98}

\begin{thm}\label{thm:deadlock}
	Let $V = \Seq{C_1,\ldots,C_n}$ be a network. In a deadlock state $\sigma$:
	\begin{enumerate}
		\item Each $C_i$ must be blocked.
		\begin{itemize}
			\item A process $C_i$ is blocked in $\sigma$ if $A_i \subseteq Refusals(\sigma)$.
		\end{itemize}
		\item There must be a cycle of ungranted requests among components.
		\begin{itemize}
			\item Given a state $\sigma$, a cycle of ungranted requests is a sequence $c \in \{1\ldots n\}^*$ such that for all $i$ in $\{1\ldots |c|\}$,  $c_i \ur c_{i \oplus 1}$ holds, where $\oplus$ denotes addition modulo $|c|$ and $|c|$ is the length of sequence/cycle $c$.
		\end{itemize}
	\end{enumerate}
\end{thm}

In this paper, we will mainly prove that a system is deadlock free by showing that a cycle of ungranted requests cannot arise in any conceivable state of the system. Since such a cycle is a necessary condition for a deadlock, deadlock freedom can be proved this way. We finish this section by illustrating a few of the concepts we have presented.

\begin{example} We discuss the concepts of communication and snapshot graph using the example of the symmetric (deadlocking) dining philosophers. This system is very similar to the asymmetric version that we have defined in Running Example~\ref{ex:philosophers} but instead of having one right-handed philosopher (as captured in component $APhil$), all philosophers are left-handed (as in component $Phil)$. We discuss an instance of this system with 3 (left-handed) philosophers and 3 forks. Since all philosophers are left-handed, they first acquire their left-hand-side fork in order to eat. If all of them acquire their left-hand-side fork at the same time, let us call this system state $\sigma$, then all forks have been acquired and none of them can acquire their right-hand-side one; a deadlock occurs. 
	
Figure~\ref{fig:comm_snap_graph} depicts the communication graph of this system (left-hand side) and the snapshot graph for system state $\sigma$ (right-hand side). On the snapshot graph, a dependency from $Fork_x$ to $Phil_y$ arises because the philosopher $y$ has acquired fork $x$ but has not released it yet. So, the fork is offering event $putdown.x.y$ which is being refused by the philosopher, who is trying to acquire its right-hand-side fork. A dependency from $Phil_x$ to $Fork_y$ occurs because the philosopher $x$ is trying to acquire the fork $y$, which has already been acquired by the philosopher who is next in the cycle of dependencies ($Phil_{x\oplus1}$, where $\oplus$ is addition modulo 3). The cycle of ungranted requests in the snapshot graph is an evidence of the deadlock system state $\sigma$ represents. Note that ungranted requests can only arise (in a snapshot graph) between components that are connected by an edge in the communication graph; if two components do not share an event, there cannot be an ungranted request between them. \qedb{}
	
	\begin{figure}[t]
		\centering
		\begin{minipage}{0.35\textwidth}
		\resizebox{\textwidth}{!}{%
		\begin{tikzpicture}[shorten >=1pt,node distance=1cm]
		\node[sqstate]	(p0)                {$Phil_0$};
		\node[sqstate]	(f0) [above=of p0] {$Fork_0$};
		\node[sqstate]	(p2) [above right=of f0] {$Phil_2$};
		\node[sqstate]	(f2) [below right=of p2] {$Fork_2$};
		\node[sqstate]	(p1) [below=of f2] {$Phil_1$};
		\node[sqstate]	(f1) [below left=of p1] {$Fork_1$};
		\path[-] 
		(p0)  edge node {} (f0)
		(f0) edge node {} (p2)
		(p2) edge node {} (f2)
		(f2) edge node {} (p1)
		(p1) edge node {} (f1)
		(f1) edge node {} (p0);
		\end{tikzpicture}}
	\caption{Communication graph for symmetric dining philosophers.}
	\end{minipage}\hspace{2cm}%
	\begin{minipage}{0.35\textwidth}
		\resizebox{\textwidth}{!}{%
			\begin{tikzpicture}[shorten >=1pt,node distance=1cm]
			\node[sqstate]	(p0)                {$Phil_0$};
			\node[sqstate]	(f0) [above=of p0] {$Fork_0$};
			\node[sqstate]	(p2) [above right=of f0] {$Phil_2$};
			\node[sqstate]	(f2) [below right=of p2] {$Fork_2$};
			\node[sqstate]	(p1) [below=of f2] {$Phil_1$};
			\node[sqstate]	(f1) [below left=of p1] {$Fork_1$};
			\path[{latex[scale=3.0]}-] 
			(p0)  edge node {} (f0)
			(f0) edge node {} (p2)
			(p2) edge node {} (f2)
			(f2) edge node {} (p1)
			(p1) edge node {} (f1)
			(f1) edge node {} (p0);
			\end{tikzpicture}}
	\end{minipage}
		\caption{Communication graph and snapshot graph for symmetric dining philosophers.}
		\label{fig:comm_snap_graph}
	\end{figure}
\end{example}

\section{Conflict freedom, acyclic networks and decomposition}
\label{sec:decomposition}

Conflict freedom can be a very helpful property in proving deadlock freedom for a system. It can be used to decompose a proof of deadlock freedom for a system or, even, to prove that an acyclic network is deadlock free. In this section, we present a refinement expression that captures conflict freedom for a pair of components. An important advantage of this formalisation is that it can be automatically checked by a refinement checker. We also discuss how this property, and our refinement expression, can be used to break down a deadlock-freedom proof and to show an acyclic network deadlock free.

\begin{defn} A \emph{conflict} is a cycle of ungranted requests of size two, i.e. a cycle between a pair of components in a network. In a system state $\sigma$, a conflict between components $i$ and $j$ arises if and only if $i \ur j$ and $j \ur i$. In such a state, both components are willing to interact with one another, but they cannot agree on the event they need to synchronise on. Then, a pair of components $i$ and $j$ is \emph{conflict free} if and only if in there is no system state in which a conflict between $i$ and $j$ occurs.	
\end{defn}

Conflict freedom can be more naturally captured by a refinement expression if the pair of components being verified is placed in a particular behavioural context. This context abstracts the behaviour of both components by using the process \cspm{Abs}. It abstracts away the events that components can perform individually as they do not play a part in making a system deadlock. This abstraction plays a fundamental role in our work; instead of their original behaviour, our behavioural analyses examine the abstract behaviour of components.

\begin{defn} \label{def:abs} For a given network $V = \Seq{C_1,\ldots,C_n}$, where $C_i = (A_i,P_i)$, we have that \cspm{Abs(i)} = $P_i \hide (\Sigma - \textit{Voc})$.
\end{defn}

Our context is also designed so it offers the special event $req$ whenever these abstract components can both offer an event from $A_i \inter A_j$. This context is given by the process \verb"Context".

\begin{defn} Let $C_i$ and $C_j$ be	two components of network $V$.
	\begin{itemize}
		\item[] \resizebox{0.92\textwidth}{!}{\texttt{Context(i,j) = Ext(i,j)[union(A(i),{req})||union(A(j),{req})] Ext(j,i)}}
	\end{itemize}
	where \resizebox{0.90\textwidth}{!}{\texttt{Ext(i,j) = Abs(i) [[ x <- x, x <- req | x <- inter(A(i),A(j))]]}}
\end{defn}

When in this context, if both components are making requests to each other (i.e. they are both offering events in $A_i \cap A_j$) but they do not agree on this event (i.e. they offer different events), then they both can offer $req$ so they can synchronise on $req$ but they cannot synchronise  on any event on $A_i \cap A_j$. So, a conflict arises when the $req$ event is offered
and $A_i \inter A_j$ is refused. Hence, a conflict free pair of
processes does not have some revival of the form
$(s, A_i \inter A_j ,req)$. So, the characteristic process \verb"ConflictFreeSpec" capturing conflict freedom must have all possible revivals but these ones.

\begin{defn} Let $C_i$ and $C_j$ be two components in network $V$.
	
\begin{verbatim}
ConflictFreeSpec(i,j) = 
    let U_A = union(A(i),A(j))
        I_A = inter(A(i),A(j))
        CF = ((|~| ev : I_A @ ev -> CF) 
                 [] req -> CHAOS(union(U_A,{req})))
             |~| 
             (|~| ev : U_A @ ev -> CF)
within CF
\end{verbatim}
	\noindent
	where:
		\verb"CHAOS(A) = SKIP |~| STOP |~| (|~| ev : A @ ev -> CHAOS(A))"
\end{defn}

The following theorem depicts the refinement expression we propose to check conflict freedom. Note that we use the stable revivals model, as this property can be more intuitively captured in this model. The reason is the nature of conflict freedom. A pair of processes are conflict free if they are not at all willing to engage or if they are willing and able to engage. This implies that in the stable failures model, we would need a process that could refuse all shared events as well as offer some events to engage, but this intuitively violates the property that refusals should be subset closed.   

\begin{thm} \label{thm:conflictFreedomAssertion}
		\verb"ConflictFreeSpec(i,j) [V= Context(i,j)" $\Rightarrow$ the pair of components $(C_i,C_j)$ is conflict free. 
\end{thm}

\begin{proof} In a conflict free state, the \verb"Context" process must not
have a revival of the form $(s, X ,req)$ where $A_i \inter A_j \subseteq X$. 
After calculation of the revivals of the \verb"ConflictFreeSpec", its revivals
are given by the following set comprehension expression $\{(s,X,a) | s \in (A_i \cup
A_j \cup \{req\})^* \land a \in (A_i \cup A_j\cup \{req\}) \land a\not \in X \land (a = req \Rightarrow
(A_i \cap A_j) \not \subseteq X) \}$; this specification has all the
possible revivals but the ones generated by a conflict. If the refinement
expression holds, then $revivals(\verb"ConflictFreeSpec(i,j)") \supseteq
revivals(\verb"Context(i,j)")$. Hence, in this case \verb"Context" has
only conflict free revivals. For the other components of this model, $deadlocks$ and $traces$,
the restrictions are evident. Traces are not restricted at all,
$traces(\verb"ConflictFreeSpec(i,j)") = (A_i \cup A_j \cup \{req\})^*$,
also as deadlock can only arise if there is a conflict, we restrict the
set of deadlocks to be empty, $deadlocks(\verb"ConflictFreeSpec(i,j)") =
\emptyset$.
\end{proof}

Conflict freedom can be used to break down the verification of deadlock for a network to the analysis of some of its subnetworks. In the communication graph of a network, the \emph{disconnecting edges} are the edges whose removal would increase the number of connected components in this graph -- these are bridges in graph-theoretic terms. We call essential subnetworks the connected components resulting after removing some of these edges.

\begin{thm}[Theorem 4 in \cite{Brookes91}] \label{thm:decomposition} A network $V$ is deadlock free if the essential subnetworks resulting from the removal of conflict-free disconnecting edges are deadlock free. A disconnecting edge is conflict free if and only if the two components participating on it are conflict free.
\end{thm}

We give an example to illustrate the concepts linked to decomposition.

\begin{example} Let $V = \Seq{C_0,C_1,C_2,C_3,C_4,C_5}$ be a live network for which communication graph is given in Figure~\ref{fig:decomposition_example}. This network has two rings ($C_0,C_1,C_2$ and $C_3,C_4,C_5$) which are interconnected via components $C_0$ and $C_3$. Also, let $\sigma_1$, $\sigma_2$, and $\sigma_3$ be states of this network for which snapshot graphs are also depicted in Figure~\ref{fig:decomposition_example}. 
	
This network has a single disconnecting edge $(C_0,C_3)$. Note that by removing this edge, we end up with two essential subnetworks (i.e. connected components in graph-theoretic terms) $\Seq{C_0,C_1,C_2}$ and $\Seq{C_3,C_4,C_5}$. If, instead, we decided to remove any other edge, we would end up with a single connected component. Hence, all other edges are not disconnecting.   
	
In a live network, a component is either blocked because it is in a path of ungranted requests leading to a blocked subnetwork or because it is in a cycle of ungranted requests; such a cycle is sort of a fundamental blocked subnetwork. Considering our network, a deadlocked state could arise because there is a conflict between our two rings, i.e. a conflict between $C_0$ and $C_3$, as for instance in state $\sigma_1$. Note that in this state, all other components depend on this pair of components to progress. If we remove the $C_0,C_3$ edge (from the communication graph)  and analyse the two rings independently, these two separate subnetworks could even be deadlock free and still admit exactly the paths of ungranted requests leading to the conflict shown in $\sigma_1$ when put together. Note that these paths on their own are not blocking either ring (hence, deadlock free could admit these paths), the conflict is the root cause of the deadlock. Therefore, inadvertently removing disconnecting edges and might lead to unknowingly removing the root cause of a deadlock from our analysis. Disconnecting edges can only be removed if they are conflict free.

Let us assume now that the edge $C_0,C_3$ is conflict free (so state $\sigma_1$ is unreachable). For a deadlock to arise, it must be that one of the rings is blocked and the components in the other ring are in ungranted-request paths leading to it. This happens, for instance, in state $\sigma_2$ where we have that the subnetwork $\Seq{C_3,C_4,C_5}$ is blocked by a cycle of dependencies and the other ring (involving $C_0,C_1,C_2$) depends on this subnetwork, so we have a deadlock. As our disconnecting edge is conflict free, we could analyse our rings independently. This state shows, however, that it only takes one blocked (essential) subnetwork to make a system deadlock. Note here that the path in $\sigma_2$ around ring $C_0,C_1,C_2$ is a valid configuration of a deadlock free (sub)network. The cycle of ungranted requests around the ring $C_3,C_4,C_5$, however, means that this subnetwork deadlocks.

If the edge $C_0,C_3$ is conflict free and the two rings are independently deadlock free, it is impossible for one ring to be blocked by the other.  State $\sigma_3$ shows a state where ring $C_0,C_1,C_2$ depends on a \emph{progressing} ring $C_3,C_4,C_5$. \qedb{}
	
	\begin{figure}[t]
			\centering
			\begin{minipage}{.45\textwidth}
		\resizebox{\textwidth}{!}{%
		\begin{tikzpicture}[shorten >=1pt,node distance=1cm]
		\node[sqstate]	(q0)                {$C_0$};
		\node[sqstate]	(q1) [above left=of q0] {$C_1$};
		\node[sqstate]	(q2) [below left=of q0] {$C_2$};
		\node[sqstate]	(q3) [right=of q0] {$C_3$};
		\node[sqstate]	(q4) [above right=of q3] {$C_4$};
		\node[sqstate]	(q5) [below right=of q3] {$C_5$};
		\path[-] 
			(q0) edge node {} (q3)
			(q0) edge  node {} (q2)
			(q2) edge  node {} (q1)
			(q1) edge  node {} (q0)
			(q3) edge  node {} (q4)
			(q4) edge  node {} (q5)
			(q5) edge node {} (q3);
		\end{tikzpicture}}
		\end{minipage}\hspace{1cm}%
		\begin{minipage}{.45\textwidth}
			\resizebox{\textwidth}{!}{%
			\begin{tikzpicture}[shorten >=1pt,node distance=1cm]
			\node[sqstate]	(q0)                {$C_0$};
			\node[sqstate]	(q1) [above left=of q0] {$C_1$};
			\node[sqstate]	(q2) [below left=of q0] {$C_2$};
			\node[sqstate]	(q3) [right=of q0] {$C_3$};
			\node[sqstate]	(q4) [above right=of q3] {$C_4$};
			\node[sqstate]	(q5) [below right=of q3] {$C_5$};
			\path[-{latex[scale=3.0]}] 
			(q0) edge [bend left]  node {} (q3)
			(q2) edge node {} (q1)
			(q1) edge node {} (q0)
			(q3) edge [bend left] node {} (q0)
			(q4) edge node {} (q5)
			(q5) edge node {} (q3);
			\end{tikzpicture}}
		\end{minipage} \vspace*{1cm} \\ 
	\begin{minipage}{.45\textwidth}
		\resizebox{\textwidth}{!}{%
			\begin{tikzpicture}[shorten >=1pt,node distance=1cm]
			\node[sqstate]	(q0)                {$C_0$};
			\node[sqstate]	(q1) [above left=of q0] {$C_1$};
			\node[sqstate]	(q2) [below left=of q0] {$C_2$};
			\node[sqstate]	(q3) [right=of q0] {$C_3$};
			\node[sqstate]	(q4) [above right=of q3] {$C_4$};
			\node[sqstate]	(q5) [below right=of q3] {$C_5$};
			\path[-{latex[scale=3.0]}] 
			(q0) edge node {} (q3)
			(q2) edge  node {} (q1)
			(q1) edge  node {} (q0)
			(q3) edge  node {} (q4)
			(q4) edge  node {} (q5)
			(q5) edge  node {} (q3);
			\end{tikzpicture}}
	\end{minipage}\hspace{1cm}%
	\begin{minipage}{.45\textwidth}
		\resizebox{\textwidth}{!}{%
			\begin{tikzpicture}[shorten >=1pt,node distance=1cm]
			\node[sqstate]	(q0)                {$C_0$};
			\node[sqstate]	(q1) [above left=of q0] {$C_1$};
			\node[sqstate]	(q2) [below left=of q0] {$C_2$};
			\node[sqstate]	(q3) [right=of q0] {$C_3$};
			\node[sqstate]	(q4) [above right=of q3] {$C_4$};
			\node[sqstate]	(q5) [below right=of q3] {$C_5$};
			\path[-{latex[scale=3.0]}] 
			(q0) edge node {} (q3)
			(q2) edge  node {} (q1)
			(q1) edge  node {} (q0)
			(q3) edge  node {} (q4)
			(q4) edge  node {} (q5);
			\end{tikzpicture}}
	\end{minipage}
		\caption{Communication graph and snapshot graphs for states $\sigma_1$, $\sigma_2$, and $\sigma_3$, respectively (left to right, top to bottom), for our examples. }
		\label{fig:decomposition_example}
	\end{figure}
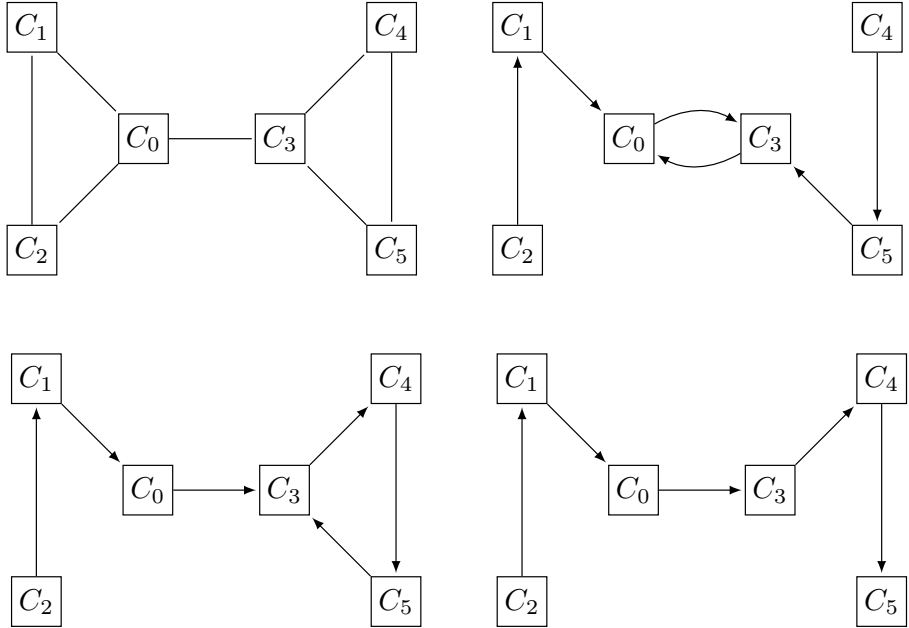 
\end{example}

Thus, our refinement expression can be used to show that a disconnecting edge is conflict free, enabling one to decompose the network into smaller essential subnetworks. Also, note that the identification of disconnecting edges can be carried out statically, i.e., by examining the communication graph, so generally this should be considerably simpler than showing conflict freedom for them. Note that a given network has a unique set of conflict-freedom disconnecting edges that can be removed to decompose the network.

In addition to that, from this theorem, we can deduce the following corollary:
\begin{cor}
	A (live) conflict-free acyclic (topology) network must be deadlock free.
\end{cor}

A network is conflict free if and only if all its edges are conflict free. Note that for an acyclic network, all edges are disconnecting ones. So, provided that all edges are conflict free, we can remove them and, as a result, we would have essential subnetworks with a single component. Thus, as components are deadlock-free, by the busyness requirement, this acyclic network must be deadlock free.

So, using our refinement expression, one can systematically decompose a network or even prove deadlock freedom for conflict-free acyclic networks. Both these applications can substantially reduce the complexity of deadlock-freedom analysis at a fairly low price; our conflict analysis only involves the examination of pairs of components as opposed to the system's overall behaviour. For instance, a conflict-free acyclic network can be simply ensured deadlock free by this sort of pairwise (local) analysis; we illustrate this with an example.

\begin{crexample}{\ref{ex:buffer}} Our ring-buffer network can be checked deadlock free by using decomposition alone. In this example, we analyse a network with one controller and three storage cells. In Figure~\ref{fig:ringbuffer}, we depict the communication graph of our example system and which sort of conflicts could potentially happen (they do not actually happen as we discuss next). $Contr$ represents the controller component, whereas $Cell_i$ depicts a $Cell(i)$ component. This system has an acyclic communication graph (i.e. topology) so every edge is disconnecting. Moreover, every edge (i.e. pair of components connected by an edge) is conflict free: whenever the controller wants to read from or write to a cell, it can do so. So, none of the possible conflicts depicted in Figure~\ref{fig:ringbuffer} can arise in any given system state. As all disconnecting edges are conflict free, we can decompose this network by removing all edges. This process results in 4 essential networks all of which have a single component. Since all components are deadlock free by virtue of our network being live, these essential subnetworks are deadlock free. Finally, by Theorem~\ref{thm:decomposition}, this network must be deadlock free.    
	\begin{figure}[t]
		\centering
		\begin{minipage}{0.3\textwidth}
			\resizebox{\textwidth}{!}{%
				\begin{tikzpicture}[shorten >=1pt,node distance=1.5cm]
				\node[sqstate]	(q0)                {$Contr$};
				\node[sqstate]	(q1) [above right=of q0] {$Cell_1$};
				\node[sqstate]	(q2) [right=of q0] {$Cell_2$};
				\node[sqstate]	(q3) [below right=of q0] {$Cell_3$};
				\path[-] 
				(q0) edge node {} (q1)
				(q0) edge  node {} (q2)
				(q0) edge  node {} (q3);
				\end{tikzpicture}}
	\end{minipage} \hspace{3cm}%
	\begin{minipage}{0.3\textwidth}
			\resizebox{\textwidth}{!}{%
			\begin{tikzpicture}[shorten >=1pt,node distance=1.5cm]
			\node[sqstate]	(q0)                {$Contr$};
			\node[sqstate]	(q1) [above right=of q0] {$Cell_1$};
			\node[sqstate]	(q2) [right=of q0] {$Cell_2$};
			\node[sqstate]	(q3) [below right=of q0] {$Cell_3$};
			\path[-{latex[scale=3.0]}] 
			(q0) edge [bend left=10, dotted] node {} (q1)
			(q1) edge [bend left=10, dotted] node {} (q0)
			(q0) edge [bend left=10, dotted]  node {} (q2)
			(q2) edge [bend left=10, dotted] node {} (q0)
			(q0) edge [bend left=10, dotted]  node {} (q3)
			(q3) edge [bend left=10, dotted]  node {} (q0);
			\end{tikzpicture}}
	\end{minipage}
	\caption{Communication graph for RingBuffer network with 3 storage cells and example of conflicts that could arise, respectively.}
	\label{fig:ringbuffer}
\end{figure}
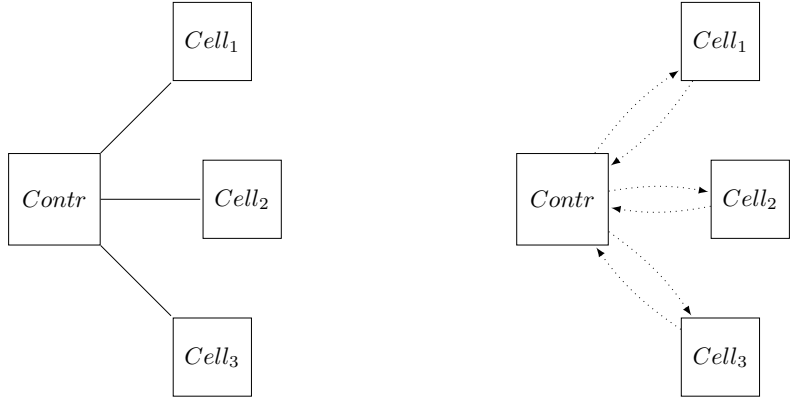
\end{crexample}


\section{Behavioural patterns}
\label{sec:patterns}

Despite being useful, conflict-freedom testing has its limitation. For instance, it is unable to show deadlock freedom for cyclic-topology systems or even to decompose systems with no disconnecting edges. For these cases, we propose pattern adherence as an alternative effective verification technique that relies on local (compositional) analysis to ensure deadlock freedom. Once again, we give up completeness in exchange for efficiency. We can only ensure deadlock freedom for systems that adhere to one of the communication/behavioural patterns that we propose but adherence to these patterns can be efficiently tested in a local/compositional way. 

In this section, we present a characterisation of behavioural patterns using refinement expressions that can prove deadlock freedom for some networks with an arbitrary communication topology. We introduce our formalisation and a proof of their soundness. We formalise requirements on the behaviour of components as refinement assertions. This formalisation permits automatic
checking of behavioural constraints using refinement checkers, providing
their model sizes are tractable.

\subsection{Resource allocation}
\label{sec:resourceAllocation}

The resource allocation pattern can be applied to systems that, in order to perform an action, have to acquire some shared resources. In this
pattern, the components of a network are divided into users and resources. A user represents a component of the system that needs to acquire
some resources in order to fulfil its final purpose. A resource is at the disposal of the users of the system.

As with design patterns for object-oriented programming languages, our behavioural patterns are specified in terms of some distinguished elements. For instance, when designing a resource allocation network, some components are meant to be users, whilst others are meant to be resources.
These pattern elements are identified through a \textit{pattern descriptor}.

A resource-allocation descriptor for a network $V$, with $n$ components and $\Sigma$ as alphabet, is a tuple $\mathcal{M} = (\mathcal{C},  acquire, release)$ containing a set $\mathcal{C} \subseteq \{1\ldots n\} \times \{1\ldots n\}$ and two functions $acquire$ and $release$. Each pair $(i,j) \in \mathcal{C}$ represents the existence of a connection in $V$ between the user component $i$ and the resource component $j$. The function application $acquire(i,j)$ ($release(i,j)$) gives the event used by $i$ to acquire (release) $j$. These functions must be defined to all pairs in $\mathcal{C}$. As conventions, $users \defs \{ i \mid \exists j : \{1\ldots n\} \spot (i,j) \in \mathcal{C}\}$, $resources \defs \{ j \mid \exists i : \{1\ldots n\} \spot (i,j) \in \mathcal{C}\}$, $resources(i) \defs \{ j \mid (i,j) \in \mathcal{C}\}$, $users(j) \defs \{ i \mid (i,j) \in \mathcal{C}\}$.

A network and a resource allocation descriptor are compliant with the resource allocation pattern if they fulfil some structural and behavioural conditions. The structural constraint restricts the static elements of the network. For instance, it may restrict which connections can be made between components or which events are shared amongst components. On the other hand, behavioural constraints restrict the behaviour of the components of the network.

The structural constraint for this pattern requires the identification of elements as either resources or users. Additionally, it restricts which events are shared. In this case, only events for acquisition and release of resources can be shared. This constraint, which appears recurrently in our patterns, singles out which events are used for interaction between components. Therefore, we can restrict the behaviour of components on these events to avoid deadlocks.

\begin{defn}\label{def:resourceAllocation} Let $V = \Seq{C_1,\ldots,C_n}$ be a network where $C_i = (A_i,P_i)$, and $\mathcal{M}$ a resource allocation pattern descriptor for $V$. The network $V$ and the descriptor $\mathcal{M}$ are structurally compliant if and only if the following predicates hold.	
	
	\begin{itemize}
		\item $partitioned \defs \{1\ldots n\} = users \cup resources \land users \cap resources = \emptyset$
		\item $\begin{aligned}[t]
		&mutually\_disjoint\_events \defs \\ & \quad \lnot \exists i : users; j : resources \spot acquire(i,j) = release(i,j)
		\end{aligned}$
		\item $\begin{aligned}[t] & controlled\_alpha\_users \defs \\ & \quad \forall i : users \spot A_i \inter \textit{Voc} = \{acquire(i,j),release(i,j) \mid
		j \in resources(i)\}
		\end{aligned}$
		\item $\begin{aligned}[t] & controlled\_alpha\_resources \defs \\ & \quad \forall i : resources \spot  A_i \inter \textit{Voc} = \{acquire(j,i),release(j,i)~|~j \in users(i)\} \end{aligned}$
	\end{itemize}
\end{defn}

On the behavioural side, we present two CSP processes that define the expected behaviour of a user component and of a resource component. The resource component should offer the events of acquisition to all users able to acquire this
resource and, once acquired, it offers the release event to the user that has
acquired it.

\begin{defn} Let $V = \Seq{C_1,\ldots,C_n}$ be a network, and $\mathcal{M}$ a resource allocation descriptor for $V$. $ResourceSpec(i)$ defines the expected behaviour of a resource component. 
\begin{lstlisting}
ResourceSpec(i) = [] j : users(i) @ acquire(j,i) -> 
	release(j,i) -> Resource
\end{lstlisting}
\end{defn}

A user component should first acquire all the necessary resources, and then release
them. Both acquiring and releasing must be performed using the order denoted by the
$order(i)$ sequence.

\begin{defn} Let $V = \Seq{C_1,\ldots,C_n}$ be a network, $\mathcal{M}$ a resource allocation descriptor for $V$, and $order(i)$ a function giving the sequence in which resources are acquired by component $i$. $UserSpec(i)$ defines the expected behaviour of a user component. 
\begin{lstlisting}
UserSpec(i) =    
	let Acquire(s) = 
			if s != <> then acquire(i,head(s)) -> Acquire(tail(s))
			else SKIP 
		Release(s) = 
			if s != <> then release(i,head(s)) -> Release(tail(s)) 
			else SKIP
		User(s) = Acquire(s);Release(s);User(s)
	within User(order(i))
\end{lstlisting}
\end{defn}

To ensure that a component meets its specification, the behavioural constraint requires the stable failure
refinement relation to hold between the specification and the behaviour of a component.

\begin{defn}  Let $V = \Seq{C_1,\ldots,C_n}$ be a network where $C_i = (A_i,P_i)$, $\mathcal{M}$ a resource allocation pattern descriptor for $V$, $order(i)$ a function giving the sequence in which resources are acquired by component $i$, and $>_{RA}$ a strict total order on resources. $V$ and $\mathcal{M}$ are behaviourally compliant if and only if the following hold.
	\begin{itemize}
		\item $\forall i : users \bullet \texttt{ UserSpec(i) [F= } \cspm{Abs(i)}$
		\item $\forall i : resources \bullet \texttt{ ResourceSpec(i) [F= } \cspm{Abs(i)}$
		\item $\forall i : users \spot order(i) \text{ must respect }>_{RA}$
	\end{itemize}
A sequence $\Seq{s_1,\ldots,s_n}$ \emph{respects} an order $>$ if the elements in the sequence are ordered respecting $>$, namely, for all $i \in \{1\ldots n\}$ we have that $s_i > s_{i+1}$.
\end{defn}

Note that we require an abstract version of a component's behaviour to comply with its specification. The reason is that, to guarantee deadlock freedom, we only need to regulate the behaviour related to events used in the interaction between components. The behaviour of a component on non-shared events is not relevant in deadlock analysis, as the component can perform them individually. So, in the analysis of deadlock freedom, we can study the network composed of the abstract behaviours of components, rather than the fully detailed network. This result is presented in the following lemma.

\begin{lem} Let $V = \Seq{C_1,\ldots,C_n}$ be a network where $C_i = (A_i,P_i)$, and $V' = \Seq{C'_1,\ldots,C'_n}$ another network where $C'_i = (A_i,\verb|Abs(i)|)$; \verb|Abs(i)| as per Definition~\ref{def:abs}. If $V'$ is deadlock free then so is $V$. 
	\label{lem:abs}
\end{lem}

\begin{proof} We prove this claim by contradiction. Assuming that $V'$ is deadlock free and $V$ is not, we reach a contradiction.
	Let us assume that $\sigma = (s,(R_1,\ldots,R_n))$ is a deadlock state of $V$, thus $Refusal(\sigma) = \Sigma$. In $\sigma$, none of the components of $V$ must be willing to perform an event that is not in the vocabulary, that is, $\overline{\textit{Voc}} \subseteq R_i$ for all $i \in \{1\ldots n\}$. If that was not the case, then $\sigma$ would not be a deadlocked state. Hence, from the definition of a network state and from the clause calculating the failures for the hiding operator, we can deduce that the state $\sigma' = (s \restrict \textit{Voc}, (R_1,\ldots,R_n))$ is a valid state for $V'$. So, since $Refusal(\sigma) = Refusal(\sigma')$ and both networks have the same alphabet, then $\sigma'$ represents a deadlock for $V'$, thus a contradiction.
\end{proof}

As the main result of this section, we show that compliance to the resource-allocation pattern guarantees deadlock freedom. It ensures that resources in a path of ungranted requests respect our strict order $>_{RA}$, namely, if there is a path from $r_1$ to $r_n$ then $r_1 >_{RA} r_n$. Hence, a cycle of ungranted requests would lead to a contradiction in the form of $r >_{RA} r$. Therefore, such cycles cannot arise and that, in turn, guarantees deadlock freedom. This sort of coincidence between paths of ungranted requests and a strict order is a core common idea shared by our patterns which makes them sound. Note that the idea of ordering resources and their acquisition to prevent deadlocks, which inspired ours and many other works, reaches back decades~\cite{Coffman71,Dijkstra71}.

\begin{thm} \label{thm:ra} Let $V = \Seq{C_1,\ldots,C_n}$ be a network where $C_i = (A_i,P_i)$, $\mathcal{M}$ a resource allocation pattern descriptor for $V$, $order(i)$ a function giving the sequence in which resources are acquired by component $i$, and $>_{RA}$ a strict total order on resources. If $V$ and $\mathcal{M}$ are resource allocation compliant then $V$ is deadlock free.
\end{thm}

\begin{proof} We prove this theorem by showing that the network $V' = \Seq{C'_1,\ldots,C'_n}$, where $C'_i = (A_i,\verb|Abs(i)|)$, is deadlock free and by using Lemma \ref{lem:abs}. 
	
	To prove that $V'$ is deadlock free, we rely on the second condition of Theorem \ref{thm:deadlock}. So, we show that there cannot be a cycle of ungranted requests between components of this network.
	
	First, given that $partitioned$ holds, we know that a component must be either a resource or a user. Moreover, thanks to $mutually\_disjoint\_events$, we know that events cannot be used for both acquiring and releasing a resource. Conditions $controlled\_alpha\_users$, $controlled\_alpha\_resources$ and triple-disjointness implies that no two resources, nor two users, can share an event. As no two resources, nor two users, can share an event, the predicate $\req$ cannot be met and, as a consequence, there cannot be an ungranted request between such two elements. Thus, a cycle of ungranted requests in this network must be composed of alternating users and resources. So, we move on to analyse the interaction between a user and a resource.
	
	Let $C_r$ be a resource component and $C_u$ a user one. From the required behaviour compliance, we know that the \verb|Abs(i)| has to behave exactly as \texttt{User(u)} or \texttt{Resource(r)} for $i = u$ or $i = r$, respectively. So, we can analyse the behaviour of \verb|Abs(i)| in terms of the behaviour of these two processes.
	
	Based on the behaviour of \texttt{User(u)} and \texttt{Resource(r)}, we know that an ungranted request can only arise from $u$ to $r$ in a state $\sigma$ if and only if $u$ is ready to acquire $r$, but $r$ has already been acquired by another user. In all other cases, $u$ and $r$ can successfully interact preventing the ungranted request. Note, then, that a cycle of ungranted requests can only involve resources that have already been acquired. Thus, we only discuss paths and cycles of ungranted request where all resources have been acquired. Additionally, based on \texttt{User(u)}'s behaviour, we know that (i) $u$ is trying to acquire a resource that is higher, considering $>_{RA}$, than any of its acquired resources.
	
	Two kind of ungranted requests can happen from a resource $r$ to one of its users $u$. An ungranted request from $r$ to $u$ might arise if either $r$ has not been acquired yet or $r$ has been acquired by $u$ but $u$ is not yet ready to release it. We are only interested in the later since the former case cannot be part of a cycle of ungranted requests; note that a free resource cannot be the target of an ungranted request, as the user issuing the request to acquire this resource would just be able to do so (i.e. the request would be ``granted").  
	
	So, we have that cycles of ungranted requests can only be formed by chains of the form $r \ur u \ur r'$ where $r$ has been acquired by $u$ and $u$ is trying to acquire $r'$. Such a chain implies that $r >_{RA} r'$ by (i). So, for any pair of resources $r_1$ and $r_n$ in a path of ungranted requests $r_1 \ur u_1 \ur r_2 \ur \ldots \ur r_{n-1} \ur u_{n-1} \ur r_n$, it must hold that $r_1 >_{RA} r_n$.  
	
	Note, then, that the existence of a cycle of ungranted requests would lead to a contradiction. Such a cycle is a resource-user chain of ungranted requests that begins and ends in the same resource. That would imply that a reflexive pair $(r,r)$ belongs to $>_{RA}$, contradicting the fact that $>_{RA}$ is a strict order on resources. In Figure~\ref{fig:intuitionRA}, we illustrate the coinciding of the order in which resources appear in paths of ungranted requests with the order $>_{RA}$, and the contradiction it leads to in the context of cycles of ungranted requests. 
	
	\begin{figure}[t]
		\centering
		\begin{minipage}{0.32\textwidth}
			\resizebox{\textwidth}{!}{%
				\begin{tikzpicture}[shorten >=1pt,node distance=.4cm]
				\node[sqstate, inner sep=0cm]	(r1)                {$\ r_{1}\ $};
				\node[sqstate, inner sep=0cm]	(u1) [below=of r1] {$\ u_1\ $};
				\node[sqstate, inner sep=0cm]	(r2) [below=of u1] {$\ r_2\ $};
				\node	(u2) [below=of r2] {$\vdots$};
				\node[sqstate, inner sep=0cm]	(r3) [below=of u2] {$r_{n-1}$};
				\node[sqstate, inner sep=0cm]	(u3) [below=of r3] {$u_{n-1}$};
				\node[sqstate, inner sep=0cm]	(r4) [below=of u3] {$\ r_n\ $};
				\path[-{latex[scale=3.0]}] 
				(r1)  edge node {} (u1)
				(u1) edge node {} (r2)
				(r2) edge node {} (u2)
				(u2) edge node {} (r3)
				(r3) edge node {} (u3)
				(u3) edge node {} (r4);
				
				\draw [decorate,decoration={brace,amplitude=10pt,raise=8pt},yshift=0cm]
				(r1.north east) -- (r4.south east) node [midway,right,xshift=.8cm] {implies $r_1 >_{RA} r_n$};
				\end{tikzpicture}}
		\end{minipage}\hfill%
		\begin{minipage}{0.5\textwidth}
		\resizebox{\textwidth}{!}{%
			\begin{tikzpicture}[shorten >=1pt,node distance=0.8cm]
			\node[sqstate, inner sep=0cm]	(r1)                {$\ r_1\ $};
			\node[sqstate, inner sep=0cm]	(u1) [above=of r1] {$\ u_1\ $};
			\node[sqstate, inner sep=0cm]	(r2) [above right=of u1] {$\ r_2\ $};
			\node	(u2) [below right=of r2] {$\vdots$};
			\node[sqstate, inner sep=0cm]	(rn) [below=of u2] {$\ r_n\ $};
			\node[sqstate, inner sep=0cm]	(un) [below right=of r1] {$\ u_n\ $};
			\path[-{latex[scale=3.0]}] 
			(r1)  edge node {} (u1)
			(u1) edge node {} (r2)
			(r2) edge node {} (u2)
			(u2) edge node {} (rn)
			(rn) edge node {} (un)
			(un) edge node {} (r1);
			
			\draw [decorate,decoration={brace,amplitude=10pt,raise=1.2cm},yshift=0cm] 
			(r2.north east) -- (un.south east) node [midway,right,xshift=1.4cm] {$\begin{array}[t]{l}\text{implies }{r_1 >_{RA} r_1} \\\text{(contradiction)} \end{array}$};
			\end{tikzpicture}}
		\end{minipage}
		\caption{Chain of ungranted requests and cycle of ungranted requests coincidence with $>_{RA}$.}
		\label{fig:intuitionRA}
	\end{figure}
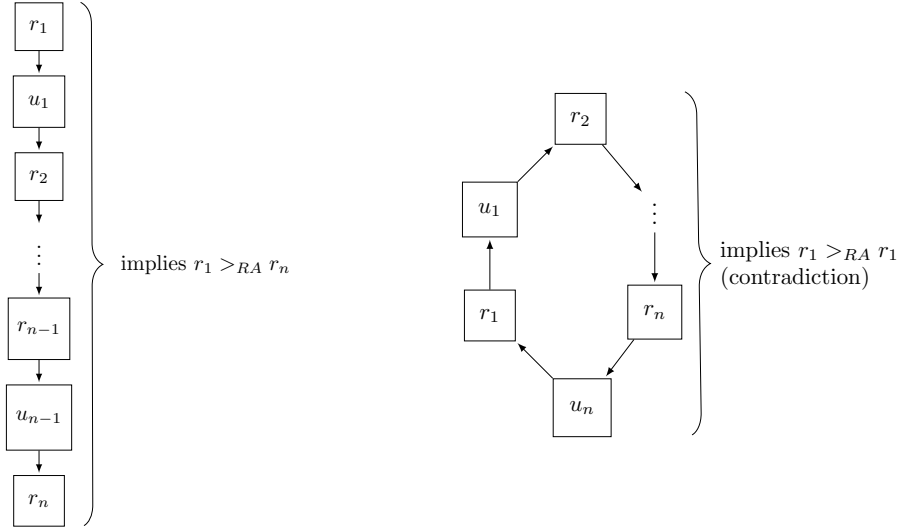
\end{proof}

We use our asymmetric dining philosophers example to illustrate how patterns can be applied to ensure deadlock freedom.

\begin{crexample}{\ref{ex:philosophers}} Our asymmetric-dining-philosophers network does not have an acyclic topology. It is in fact a large ring of alternating fork and philosopher components. Hence, decomposition is not an option as all edges are not disconnecting. Note that by removing any edge the number of connected (graph-theoretic) components/essential subnetworks does not increase; we always have a single subnetwork. In this example, we analyse a network with 3 forks and 3 philosophers. Figure~\ref{fig:diningphilosopher} depicts the communication graph of our example system and an example of a system state $\sigma$ which we discuss next. $Phil_i$ represents component $Phil(i)$, $Fork_i$ component $Fork(i)$, and $APhil_i$ component $APhil(i)$.
	
As decomposition is not an option, we apply a pattern to the entire network: the resource allocation pattern. Philosophers are users and forks are resources, and philosophers have to acquire forks according to the expected order on their indexes; this is the $>_{RA}$ order. As this network adheres to this pattern, in a cycle of ungranted requests the resources present in this cycle must have been acquired and the way in which they are ordered must respect their natural index order. For instance, assume that $\sigma$ is a network state exhibiting a cycle of ungranted requests such as the one in Figure~\ref{fig:diningphilosopher}, we explain how the network cannot reach such a state. Note that there is an ungranted request from $Fork_2$ to $APhil_2$ and from $APhil_2$ to $Fork_0$. Such a configuration can only happen if the $Fork_2$ has not been acquired by $APhil_2$, according to the behavioural requirements over users and resources enforced by our pattern. Hence, $Phil_1$ cannot have an ungranted request to $Fork_1$ as it is free to be acquired.	If all the resources were acquired, an ungranted-request path from $r_1$ to $r_2$ would coincide with our $>_{RA}$ order. Therefore, a cycle of ungranted requests could not arise as it would violate the irreflexiveness of $>_{RA}$. 
	\begin{figure}[t]
		\centering
		\begin{minipage}{0.32\textwidth}
			\resizebox{\textwidth}{!}{%
				\begin{tikzpicture}[shorten >=1pt,node distance=1cm]
				\node[sqstate]	(p0)                {$Phil_0$};
				\node[sqstate]	(f0) [above=of p0] {$Fork_0$};
				\node[sqstate]	(p2) [above right=of f0] {$APhil_2$};
				\node[sqstate]	(f2) [below right=of p2] {$Fork_2$};
				\node[sqstate]	(p1) [below=of f2] {$Phil_1$};
				\node[sqstate]	(f1) [below left=of p1] {$Fork_1$};
				\path[-] 
				(p0)  edge node {} (f0)
				(f0) edge node {} (p2)
				(p2) edge node {} (f2)
				(f2) edge node {} (p1)
				(p1) edge node {} (f1)
				(f1) edge node {} (p0);
				\end{tikzpicture}}
		\end{minipage}\hspace{2cm}%
		\begin{minipage}{0.32\textwidth}
			\resizebox{\textwidth}{!}{%
				\begin{tikzpicture}[shorten >=1pt,node distance=1cm]
				\node[sqstate]	(p0)                {$Phil_0$};
				\node[sqstate]	(f0) [above=of p0] {$Fork_0$};
				\node[sqstate]	(p2) [above right=of f0] {$APhil_2$};
				\node[sqstate]	(f2) [below right=of p2] {$Fork_2$};
				\node[sqstate]	(p1) [below=of f2] {$Phil_1$};
				\node[sqstate]	(f1) [below left=of p1] {$Fork_1$};
				\path[{latex[scale=3.0]}-] 
				(p0)  edge node {} (f0)
				(f0) edge node {} (p2)
				(p2) edge node {} (f2)
				(f2) edge [dotted] node {} (p1)
				(p1) edge node {} (f1)
				(f1) edge node {} (p0);
				\end{tikzpicture}}
		\end{minipage}
		\caption{Communication graph and snapshot graph for asymmetric dining philosophers.}
		\label{fig:diningphilosopher}
	\end{figure}
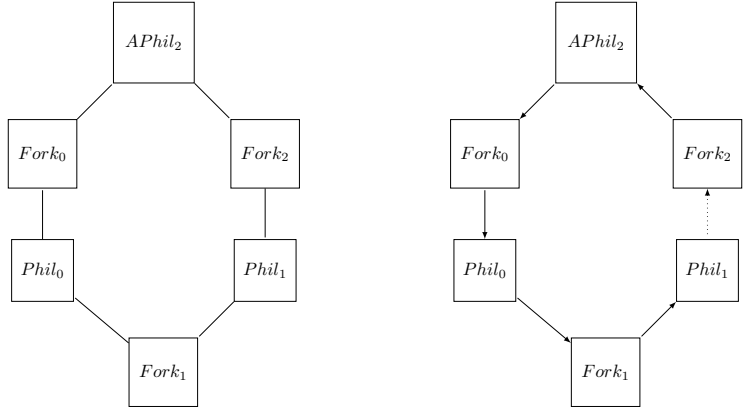
\end{crexample}

\subsection{Client/Server}

The client/server pattern applies to some networks implementing a client/server interaction architecture. In such a network, a component might behave as both a server and a client. As a server, it waits for a request from a client. As a client, it contacts a server component in the search for some service. The distinction between behaving as a server or as a client is based on the offer of events by a component. In a server state it must be offering all its server events, whereas in a client state it must be willing to request some service. The distinction between such events, as well as the identification of other elements of this pattern, is made via a pattern descriptor.

A client/server descriptor for a network $V$, with $n$ components and $\Sigma$ as alphabet, is a tuple $\mathcal{M} = (\mathcal{C}, request, responses)$ containing a set $\mathcal{C} \subseteq \{1\ldots n\} \times \{1\ldots n\}$ and functions $requests$ and $responses$. Each pair $(i,j) \in \mathcal{C}$ represents the existence of a connection in $V$ between components $i$ and $j$ such that $i$ acts as a client and $j$ as a server. The function application $requests(i,j)$ yields a set of events for which the client $i$ requests some service of the server $j$. For the request event $k$, the expected responses are given by the events in $responses(k)$. As conventions, we have that:

\begin{itemize}
	\item $\clr(i) \defs \Union \{ requests(i,j) \mid (i,j) \in \mathcal{C}\}$;
	\item $\svr(i) \defs \Union \{ requests(j,i) \mid (j,i) \in \mathcal{C}\}$;
	\item $\clp(i) \defs \bigcup \{ responses(k) \mid k \in \clr(i)\}$;
	\item $\svp(i) \defs \bigcup \{ responses(k) \mid k \in \svr(i)\}$.
\end{itemize}

A network $V$ and a client/server descriptor are structurally compliant if they fulfil some conditions. Roughly speaking, these conditions ensure that there can only be interaction between components through the use of the controlled events, that is, request and response events. Furthermore, we impose that the client/server relation between components $\mathcal{C}$ should respect a strict order on component identifiers.

\begin{defn}  Let $V = \Seq{C_1,\ldots,C_n}$ be a network where $C_i = (A_i,P_i)$, and $\mathcal{M}$ a client/server pattern descriptor for $V$, and $>_{CS}$ a strict total order on component identifiers. $V$ and $\mathcal{M}$ are structurally compliant if and only if the following predicates hold.
	\begin{itemize}
		\item $disjoint\_events \defs requests \inter responses = \emptyset$
		\begin{itemize}
			\item $requests \defs \Union \{ requests(i,j) \mid (i,j) \in \mathcal{C}\}$
			\item $responses \defs \Union\{ responses(k) \mid k \in requests\}$
		\end{itemize}
		\item $ \begin{aligned}[t] & controlled\_alpha \defs \\
		& \qquad \forall i : \{1 \ldots n\} \spot A_i \inter \textit{Voc} = \svr(i) \union \clr(i)\\
		& \qquad \qquad \union \svp(i) \union \clp(i) \end{aligned} $
		\item $ordered$ holds if and only if the relation $\mathcal{C}$ respects $>_{CS}$.
	\end{itemize}
\end{defn}

We propose two expected behaviours for a component in a client/server network. The first one concerns how it behaves as a server. When a component is behaving as a server, namely, offering some server request event, it must offer all its server request events. This is to say, as a server, a component cannot choose which requests it is able to do, but it should rather offer all its services for its clients. 

\begin{defn} Let $V = \Seq{C_1,\ldots,C_n}$, and $\mathcal{M}$ a client/server pattern descriptor for $V$. The server request specification for component $i \in \{1 \ldots n\}$ is given by the following process.
	
\begin{lstlisting}
ServerRequestsSpec(i) = 
	let sEvts = server_requests(i)
		otherEvts = diff(A(i),sEvts)
		Server = 
			((|~| ev : otherEvts @ ev -> SKIP) 
			|~| 
			([] ev : sEvts @ ev -> SKIP)) ; Server
	within if not empty(otherEvts) then Server else RUN(sEvts)
\end{lstlisting}
where \verb"RUN(evts) = [] ev : etvs @ ev -> RUN(evts)"
	
\end{defn}
In the definition of the process \cspm{ServerRequestsSpec}, we check whether the
set of non server request events is empty, since the replicated
internal choice operator is not defined for an empty set of elements. 


The second behavioural imposition restricts the request-response behaviour of components. A process, conforming to the client/server pattern, must recursively offer its request events and then the appropriate responses for the selected request event. The specification of this behaviour is given by the following process, which also has to deal with the replicated internal choice undefinedness for the empty set.

\begin{defn} Let $V = \Seq{C_1,\ldots,C_n}$, and $\mathcal{M}$ a client/server pattern descriptor for $V$. The request-response specification for component $i \in \{1 \ldots n\}$ is given by the following process.
	
\begin{lstlisting}
RequestsResponsesSpec(i) = 
	let cEvts = client_requests(i)
		sEvts = server_requests(i)
		ClientRequestsResponsesSpec = 
			|~| ev : cEvts @ ev -> 
				(if empty(responses(ev)) then SKIP
				else ([] res : responses(ev) @ res -> SKIP))
		ServerRequestsResponsesSpec = 
			|~| ev : sEvts @ ev -> 
				(if empty(responses(ev)) then SKIP 
				else (|~| res : responses(ev) @ res -> SKIP))
		C = ClientRequestsResponsesSpec; C
		S = ServerRequestsResponsesSpec; S
		CS = (ClientRequestsResponsesSpec 
			 |~| ServerRequestsResponsesSpec); CS
	within
		if empty(cEvts) and empty(sEvts) then STOP
		else 
			if empty(cEvts) then S
			else 
				if empty(sEvts) then C
				else CS
\end{lstlisting}
	
\end{defn}

We use the revivals' refinement relation to check conformance of a component's behaviour to the process \cspm{ServerRequestsSpec}. The reason is that the specification that ``either all server-requests are offered or none of them is" cannot be intuitively represented by a characteristic process in the stable failures model. Note that, intuitively, such a characteristic process would require failures that are not prefix closed. On the other hand, this process can be simply captured, in the stable revivals model, by the aforementioned characteristic process. The other specification does not suffer from this problem and can be simply captured in the stable failures model.

\begin{defn} Let $V = \Seq{C_1,\ldots,C_n}$ be a network where $C_i = (A_i,P_i)$, and $\mathcal{M}$ a client/server pattern descriptor for $V$. $V$ and $\mathcal{M}$ are behaviourally compliant if and only if the following predicates hold.
	\begin{itemize}
		\item $\forall i : \{1 \ldots n\} \spot\texttt{ ServerRequestsSpec(i) [V= } \cspm{Abs(i)}$
		\item $\forall i : \{1 \ldots n\} \spot \texttt{ RequestResponsesSpec(i) [F= } \cspm{Abs(i)}$ 
	\end{itemize}
\end{defn}

As with the previous pattern, we benefit from the order imposed on the client-server relation to show that a cycle of ungranted requests cannot arise and, as a result, a network compliant to this pattern is deadlock free.

\begin{thm}  \label{thm:cs} Let $V = \Seq{C_1,\ldots,C_n}$ be a network where $C_i = (A_i,P_i)$, and $\mathcal{M}$ a client/server pattern descriptor for $V$, and $>_{CS}$ a strict order on component identifiers. If $V$ and $\mathcal{M}$ are structural and behaviourally compliant then $V$ is deadlock free.
	\label{thm:CSdeadlockfree}
\end{thm}

\begin{proof} We prove this theorem by showing that the network $V' = \Seq{C'_1,\ldots,C'_n}$, where $C'_i = (A_i,\verb|Abs(i)|)$, is deadlock free and by using Lemma \ref{lem:abs}. 
	
	To prove the former claim, we rely on the second condition of Theorem \ref{thm:deadlock}. So, we show that there cannot be a cycle of ungranted requests between components of this network. From the validity of $mutually\_disjoint\_events$, we know that events cannot be used for both requesting and responding. 
	
	From the behavioural compliance of the network to the client/server pattern, we know that a component might be in one of three cases: ready to request as a client, waiting for a request as a server, and ready to respond.
	
	In the case a component is responding, due to behavioural compliance, it can only be willing to communicate with its peer, namely, the component that has shared a request event with it. In this case, both must be willing to engage in a shared event. The component behaving as server has to offer at least a response, whereas the client component must be waiting for any response event. Hence, in such a state, a component and its peer cannot be part of a cycle of ungranted requests, as the $\ung$ predicate does not hold for them. So, a cycle of ungranted requests can only be formed by a combination of client-requesting and server-waiting-for-request components.
	
	Given two components $i$ and $j$, there cannot be an ungranted request $i \ur j$ in a state where $i$ is a client-requesting and $j$ a server-waiting. The reason is that $j$ would be willing to engage on the request offered by $i$. So, this fact implies that a cycle of ungranted requests can only exists if all components are behaving either as a server-waiting or as a client-requesting. 
	
	So, let us first assume that a cycle involving only client-requesting components exists. This means that for each pair of adjacent elements $i$ and $j$ in the cycle, $(i,j) \in \mathcal{C}$ and consequently (by $ordered$) $i >_{CS} j$ must hold. Thus, we reach a contradiction as $ >_{CS} $ is a strict order and, based on the cycle, we can establish that a reflexive pair exists. In the case of all server-waiting components, one can use the same argument, but using the order dual to $ >_{CS} $, to reach a contradiction. Figure~\ref{fig:intuitionCS} illustrates the coincidence of order $>_{CS}$ and components in a path of ungranted requests, and the contradiction that a cycle would lead to.
	
	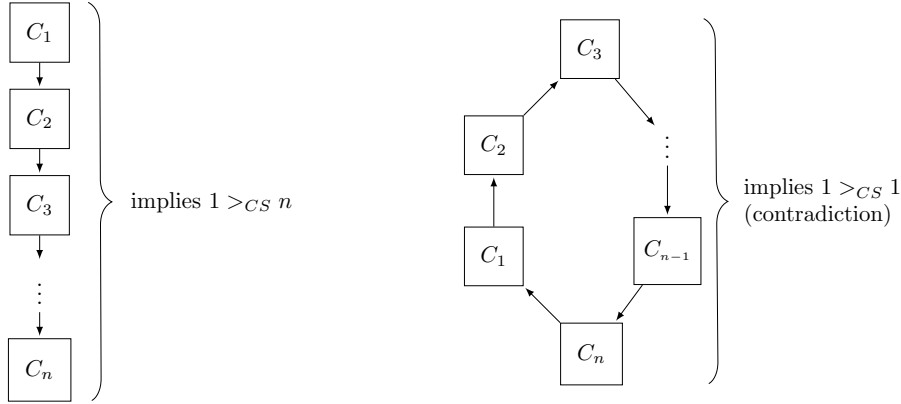
\begin{figure}[t]
		\centering
		\begin{minipage}{0.32\textwidth}
			\resizebox{\textwidth}{!}{%
				\begin{tikzpicture}[shorten >=1pt,node distance=.4cm]
				\node[sqstate, inner sep=0cm]	(c1)                {$\ C_1\ $};
				\node[sqstate, inner sep=0cm]	(c2) [below=of c1] {$\ C_2\ $};
				\node[sqstate, inner sep=0cm]	(c3) [below=of c2] {$\ C_3\ $};
				\node	(c4) [below=of c3] {$\vdots$};
				\node[sqstate, inner sep=0cm]	(cn) [below=of c4] {$\ C_n\ $};
				\path[-{latex[scale=3.0]}] 
				(c1)  edge node {} (c2)
				(c2) edge node {} (c3)
				(c3) edge node {} (c4)
				(c4) edge node {} (cn);
				
				\draw [decorate,decoration={brace,amplitude=10pt,raise=8pt},yshift=0cm]
				(c1.north east) -- (cn.south east) node [midway,right,xshift=.8cm] {implies $1 >_{CS} n$};
				\end{tikzpicture}}
		\end{minipage}\hfill%
		\begin{minipage}{0.5\textwidth}
			\resizebox{\textwidth}{!}{%
				\begin{tikzpicture}[shorten >=1pt,node distance=0.8cm]
				\node[sqstate, inner sep=0cm]	(c1)                {$\ C_1\ $};
				\node[sqstate, inner sep=0cm]	(c2) [above=of c1] {$\ C_2\ $};
				\node[sqstate, inner sep=0cm]	(c3) [above right=of c2] {$\ C_3\ $};
				\node	(c4) [below right=of c3] {$\vdots$};
				\node[sqstate, inner sep=0cm]	(cn1) [below=of c4] {$C_{\scriptscriptstyle n-1}$};
				\node[sqstate, inner sep=0cm]	(cn) [below right=of c1] {$\ C_n\ $};
				\path[-{latex[scale=3.0]}] 
				(c1)  edge node {} (c2)
				(c2) edge node {} (c3)
				(c3) edge node {} (c4)
				(c4) edge node {} (cn1)
				(cn1) edge node {} (cn)
				(cn) edge node {} (c1);
				
				\draw [decorate,decoration={brace,amplitude=10pt,raise=1.3cm},yshift=0cm] 
				(c3.north east) -- (cn.south east) node [midway,right,xshift=1.6cm] {$\begin{array}[t]{l}\text{implies }{1 >_{CS} 1} \\\text{(contradiction)} \end{array}$};
				\end{tikzpicture}}
		\end{minipage}
		\caption{Chain of ungranted requests and cycle of ungranted requests coincidence with $>_{CS}$.}
		\label{fig:intuitionCS}
	\end{figure}
\end{proof}

\subsection{Async Dynamic}

This pattern can be applied to construct networks in which participants interact via a transport layer. For instance, this pattern seems to be suited for building name-server and address-resolution systems~\cite{Plummer1982,Mockapetris87}. Participants are elements that embed the functional behaviour of the network, whereas the transport layer is a mere communication infrastructure. In such a network, a fixed number of participants, which are also known in advance, might join and leave the network. Aside from transporting messages, the transport layer also detects participants leaving and entering the network. The transport layer is composed of transport entities. These are components responsible for providing one-direction communication between two participants and detecting whether its sending participant is present or not in the network.

An async-dynamic descriptor for a network $V$, with $n$ components and $\Sigma$ as alphabet, is a tuple $\mathcal{M} = (\mathcal{C},link,send, receive,on,off,timeout)$ containing a set $\mathcal{C} \subseteq \{1\ldots n\} \times \{1\ldots n\}$, and functions $link(i,j)$, $send(i,j)$, $receive(i,j)$, $on(i,j)$, $off(i,j)$, $timeout(i,j)$. A pair $(i,j) \in \mathcal{C}$ denotes the connection from $i$ to $j$. The function $link(i,j)$ yields the transport-entity component that relay messages from $i$ to $j$. This function must be defined for all pairs in $\mathcal{C}$; $send(i,j)$ and $receive(i,j)$ denote the set of events used to pass data from $i$ to $j$; and $on(i,j)$, $off(i,j)$ and $timeout(i,j)$ denote control events that are explained later. We define $participants \defs \{ i \mid \exists j : \{1\ldots n\} \spot (i,j) \in \mathcal{C} \lor (j,i) \in \mathcal{C} \}$, $transport\_entities \defs \{ link(i,j) \mid \exists i,j : \{1\ldots n\} \spot (i,j) \in \mathcal{C}\}$. We require a given transport entity to link a unique pair of participants, so we use $source(k) = i$ and $target(k) = j$ if $link(i,j) = k$.

Structural compliance is achieved if the network's components are partitioned into transport entities and participants. In addition to that, we require the traditional shared events to be the ones controlled by the pattern.

\begin{defn} Let $V = \Seq{C_1,\ldots,C_n}$ be a network where $C_i = (A_i,P_i)$, and $\mathcal{M}$ an async-dynamic pattern descriptor for $V$, and $S(i)$ a function that gives the sequence  in which participant $i$ interacts with its peer participants. $V$ and $\mathcal{M}$ are structurally compliant if and only if the following predicates hold.
	\begin{itemize}
		\item $partitioned \defs \begin{aligned}[t]&
		participants \cap
		transport\_entities = \emptyset \\ &\land participants \cup transport\_entities = \{1\ldots n\}
		\end{aligned}$
		\item $mutually\_disjoint\_events$ holds if and only if the events used for sending, receiving, turning on, turning off and timing out are all mutually disjoint. For any two sets $X$ and $Y$, representing all the events used for two of these activities, $X \cap Y = \emptyset$ must hold;
		\item $\begin{aligned}[t]
		&controlled\_alpha\_participant \defs\\
		&\quad \forall i : participants \spot A_i \inter \textit{Voc} = \\
		& \qquad\Union \{send(i,j) \mid (i,j) \in \mathcal{C}\}\\ 
		 & \qquad \union \Union \{receive(j,i) \mid (j,i) \in \mathcal{C}\}\\ 
		 &\qquad \union \{on(i,j),off(i,j), timeout(i,j) \mid (i,j) \in \mathcal{C}\}
		\end{aligned}$
		\item $\begin{aligned}[t]
		&controlled\_alpha\_transport\_entity \defs&\\
		& \quad \forall link(i,j) : transport\_entities \spot A_{link(i,j)} \inter \textit{Voc} = \\ & \qquad send(i,j) \union receive(i,j) \union \{ on(i,j), off(i,j)\} \union \{timeout(i,j)\}&
		\end{aligned}$
	\end{itemize}
\end{defn}

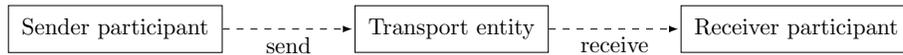
\begin{figure}[!b]
	\centering
	\resizebox{\textwidth}{!}{%
		\begin{tikzpicture}[shorten >=1pt,node distance=2cm]
		\node[rectangle, inner sep=0.2cm]	(c1)       {Sender participant};
		\node[rectangle, inner sep=0.2cm]	(c2) [right=of c1] {Transport entity};
		\node[rectangle, inner sep=0.2cm]	(c3) [right=of c2] {Receiver participant};
		\path[-{latex[scale=3.0]}] 
		(c1)  edge [dashed] node [below] {send} (c2)
		(c2) edge [dashed] node [below] {receive} (c3);
		\end{tikzpicture} }
	\caption{Illustration of the communication role performed by a transport entity.}
	\label{fig:transport_entity}
\end{figure}

On the behavioural side, we restrict the behaviours of participants and transport entities in different ways. A transport entity is expected to behave as a one-place buffer that can be overwritten with new data, providing one-direction communication as illustrated in Figure~\ref{fig:transport_entity}. In addition to that, it must be able to detect whether its sender is present or not in the network. The information about the presence of a participant is conveyed by the events on and off. If the participant is off, it means that it is no longer part of the network, it is on otherwise.

\begin{defn} Let $V = \Seq{C_1,\ldots,C_n}$ be a network where $C_i = (A_i,P_i)$, and $\mathcal{M}$ an async-dynamic pattern descriptor for $V$. The expected behaviour of the transport entity component $k$ is given by the following process.
	
\begin{lstlisting}
TransportSpec(k) =
	let i = source(k)
		j = target(k)
		On = off(i,j) -> Off 
			[] send(i,j)?data -> OnF(data)
		OnF(d) = off(i,j) -> Off 
			[] send(i,j)?data -> OnF(data)
			[] receive(i,j)!d -> On
		Off = on(i,j) -> On 
			[] timeout(i,j) -> Off
	within Off
\end{lstlisting}
\end{defn}

As mentioned, participants are the elements of the network carrying its business logic. For the purpose of deadlock analysis, we are only interested in the pattern of interaction of the participants, rather than in the business logic that they carry out. So, a participant should cyclically interact with its peer participants, first sending a message for each of its peer and then receiving messages from all of them. It might receive a timeout instead of some data, if a peer participant has left the network. At any time, a participant should be able to turn off, namely, leave the network. After leaving, the participant might re-join the network.

\begin{defn} Let $V = \Seq{C_1,\ldots,C_n}$ be a network where $C_i = (A_i,P_i)$, and $\mathcal{M}$ an async-dynamic pattern descriptor for $V$, and $S(i)$ a function that gives an order in which participant $i$ interact with its peer participants. The expected behaviour of participant $i$ is given by the following process.
	
\begin{lstlisting}
ParticipantSpec(i) =
	let	s = S(i)
		SendReceive(i,s) = 
			Send(i,s); Receive(i,s); SendReceive(i,s)
	within OnDetect(i,s);(SendReceive(i,s) /\ (SKIP |~| STOP));
		OffDetect(i,s); ParticipantSpec(i,s)
\end{lstlisting}
\end{defn}

The \verb"OnDetect" (\verb"OffDetect") process sends
a signal to inform that it is on (off) to each of the transport entity to which it
acts as a sender; this mechanism abstracts the ability
of the transport layer to detect participant status. In the same way,  The \verb"s" parameter
gives the sequence in which the participant interacts with its transport entities.
The \verb"Send" process sends messages to all transport entities that have this
participant as sender, following the order of sequence \verb"s". The \verb"Receive" process interacts with the transport entities that have
it as a receiver, also following how participants are ordered in \verb"s". This receiving interaction consists of either accepting incoming data or a timeout, in
the case that the sender associated with the transport entity in question is
off.

Note that we use the process \verb"(SKIP |~| STOP)" on the right-hand side of the interruption operator instead of, for instance, \texttt{SKIP}. The reason is that the latter construction would trivially imply deadlock freedom as a participant would be always able to turn off. On the other hand, the internal-choice construction implies that the process might not have the ability of turning off (if \cspm{STOP} is chosen), and as a consequence, one can guarantee that deadlock freedom is achieved because they are well behaved processes rather than because they can always turn off.

For a network and an async-dynamic descriptor to be behaviour compliant, participant and transport entities must meet their respective specifications. In addition to that, the sequence in which participants interact with its peers, given by $S(i)$, must not have the same component twice.

\begin{defn} Let $V = \Seq{C_1,\ldots,C_n}$ be a network where $C_i = (A_i,P_i)$, and $\mathcal{M}$ an async-dynamic pattern descriptor for $V$, and $S(i)$ a function that gives the sequence in which participant $i$ interact with its peer participants. $V$ and $\mathcal{M}$ are behaviourally compliant if and only if the following conditions hold.
	\begin{itemize}
		\item $\forall i : transport\_entities \spot \texttt{ TransportSpec [F= } \cspm{Abs(i)}$
		\item $\forall i : participants \spot \texttt{ ParticipantSpec [F= } \cspm{Abs(i)}$
		\item $\forall i : participants \spot \forall j,k : \{1 \ldots |S(i)|\} \mid j \neq k \spot S(i)_j \neq S(i)_k$
	\end{itemize}
\end{defn}

Finally, given the introduced pattern, we present the main theorem of this section. It shows that compliance to the pattern implies deadlock freedom.

\begin{thm} \label{thm:ad} Let $V = \Seq{C_1,\ldots,C_n}$ be a network where $C_i = (A_i,P_i)$, and $\mathcal{M}$ an async-dynamic pattern descriptor for $V$, and $S(i)$ a function that gives the sequence in which participant $i$ interact with its peer participants. If $V$ and $\mathcal{M}$ are behavioural and structurally compliant to the async-dynamic pattern, then $V$ is deadlock free.
\end{thm}

\begin{proof}
	From the analysis of structural restrictions, a process must be either a transport entity or a participant. This fact together with triple disjointness and the controlled-alphabet restriction imply that there can only be ungranted requests between a transport entity and a participant. To be more specific, there can only be an ungranted request between a participant and one of its sender or receiver transport entities, for a participant only shares events with these transport entities.
	

	Next, we show that there cannot be a cycle of ungranted requests in a state where all transport entities have not an \cspm{on} event as their last event. 
	
	First, we examine the behaviour of a participant $i$ when interacting with its transport entity $k$. We analyse two cases: when $i$ is a sender to $k$ and when $i$ is a receiver from $k$. When $i$ is a sender, no ungranted requests can from $i$ to $k$. Whenever $i$ is willing to communicate with $k$, $k$ is accepting a communication from $i$, be it a \cspm{send}, \cspm{on} or \cspm{off} event. When $i$ is a receiver, however, an ungranted request arises from $i$ to $k$ if $k$ is on and empty.
	
	In order to be on and empty, a transport entity $k$, linking $i$ to $j$, must have just turned on (i.e., $on(i,j)$ was its last event performed), or it must have been filled and then emptied (i.e., $receive(i,j)$ was its last event performed). In the first case, the participant $j$ has to be turning on or broadcasting data. In both cases, $j$ has to be in a state in which it can effectively communicate a send or an on event to a transport entity that has $j$ as a sender. Therefore, the network $V'$ cannot be blocked. So, we only have to establish that a cycle involving participants willing to receive messages and filled-and-then-emptied transport entities cannot arise.
	
	Let us assume that such a cycle of participants willing to receive messages and filled-and-then-emptied transport entities exist. We analyse the behaviour of a transport entity $k$, which links $i$ to $j$, and of participants $j$ and $i$.
	
	In such a cycle, $k$ must have $receive(i,j)$ as its last performed event, and based on the behaviour of a transport entity, it must have performed a $send(i,j)$ immediately before $receive(i,j)$. So, it has to have executed a trace like:
	\begin{itemize}
		\item[] $\Seq{\ldots,send(i,j),receive(i,j)}$
	\end{itemize} 
	
	Participant $j$ must be willing to receive some data from $k$. So, it has to be offering the event $receive(i,j)$. As $j$ synchronises with $k$ in $receive(i,j)$, the last occurrence of this event for $j$ and $k$ must have happened at the same time. Note that, as $receive(i,j)$ is being offered by $j$, $j$ must have broadcast between the last occurrence of $receive(i,j)$ and its current state. So, $j$ must have performed its last broadcast after the last occurrence of $receive(i,j)$.
	
	Participant $i$, as $j$, must be willing to receive some data from a transport entity. So, it must be in its receiving phase, and that means that its last broadcast has been completed. As $i$ synchronises with $k$ in $send(i,j)$, the last occurrence of this event for $i$ and $k$ must have happened at the same time.
	
	Thus, considering the behaviour of these components together, we know that $j$'s last broadcast must have started more recently than the start of $i$'s last broadcast. $j$'s last broadcast must have started after the last $receive(i,j)$ occurred. $i$'s last broadcast must have started before the last occurrence of $receive(i,j)$,  as the last occurrence of $send(i,j)$, which is part of $i$'s last broadcast, happened before $receive(i,j)$.
	
	Hence, in such a cycle, we have the following strict order being induced between participants. If $j$, $k$ and $i$ are a path in this cycle then $j$ must have had its last broadcast more recently than $i$'s last one. Let us call this order $>_{AD}$. This strict order implies that a cycle cannot happen as this would lead to a contradiction: one could deduce that a participant's last broadcast happened more recently than its last broadcast. Thus, this network is sdeadlock free. In Figure~\ref{fig:intuitionAD}, we illustrate the coinciding of these paths of ungranted requests and the order $>_{AD}$, and the contradiction that a cycle of ungranted request would lead to. 
	
	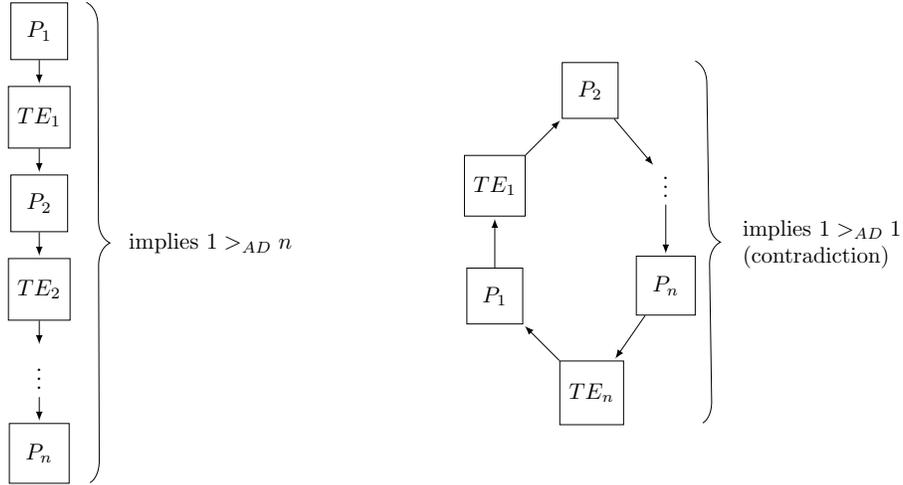
\begin{figure}[t]
		\centering
		\begin{minipage}{0.32\textwidth}
			\resizebox{\textwidth}{!}{%
				\begin{tikzpicture}[shorten >=1pt,node distance=.4cm]
				\node[sqstate, inner sep=0cm]	(c1)                {$\ P_1\ $};
				\node[sqstate, inner sep=0cm]	(c2) [below=of c1] {$TE_1$};
				\node[sqstate, inner sep=0cm]	(c3) [below=of c2] {$\ P_2\ $};
				\node[sqstate, inner sep=0cm]	(c4) [below=of c3] {$TE_2$};
				\node	(c5) [below=of c4] {$\vdots$};
				\node[sqstate, inner sep=0cm]	(cn) [below=of c5] {$\ P_n\ $};
				\path[-{latex[scale=3.0]}] 
				(c1)  edge node {} (c2)
				(c2) edge node {} (c3)
				(c3) edge node {} (c4)
				(c4) edge node {} (c5)
				(c5) edge node {} (cn);
				
				\draw [decorate,decoration={brace,amplitude=10pt,raise=8pt},yshift=0cm]
				(c1.north east) -- (cn.south east) node [midway,right,xshift=.8cm] {implies $1 >_{AD} n$};
				\end{tikzpicture}}
		\end{minipage}\hfill%
		\begin{minipage}{0.5\textwidth}
			\resizebox{\textwidth}{!}{%
				\begin{tikzpicture}[shorten >=1pt,node distance=0.8cm]
				\node[sqstate, inner sep=0cm]	(c1)                {$\ P_1\ $};
				\node[sqstate, inner sep=0cm]	(c2) [above=of c1] {$TE_1$};
				\node[sqstate, inner sep=0cm]	(c3) [above right=of c2] {$\ P_2\ $};
				\node	(c4) [below right=of c3] {$\vdots$};
				\node[sqstate, inner sep=0cm]	(cn1) [below=of c4] {$\ P_n\ $};
				\node[sqstate, inner sep=0cm]	(cn) [below right=of c1] {$TE_n$};
				\path[-{latex[scale=3.0]}] 
				(c1)  edge node {} (c2)
				(c2) edge node {} (c3)
				(c3) edge node {} (c4)
				(c4) edge node {} (cn1)
				(cn1) edge node {} (cn)
				(cn) edge node {} (c1);
				
				\draw [decorate,decoration={brace,amplitude=10pt,raise=1.2cm},yshift=0cm] 
				(c3.north east) -- (cn.south east) node [midway,right,xshift=1.6cm] {$\begin{array}[t]{l}\text{implies }{1 >_{AD} 1} \\\text{(contradiction)} \end{array}$};
				\end{tikzpicture}}
		\end{minipage}
		\caption{Chain of ungranted requests and cycle of ungranted requests coincidence with $>_{AD}$.}
		\label{fig:intuitionAD}
	\end{figure}
\end{proof}

Note that the patterns presented impose restrictions that can be efficiently checked. These are either restrictions that can be statically checked, or behavioural restrictions that can be checked by the examination of individual or pairs of processes. So, in the case of proving deadlock freedom for large systems, pattern adherence is an efficient choice and it might, in fact, be the only viable option. For example, in Section~\ref{sec:evaluation}, we present a leadership-election system, modelled after a commercial protocol, for which monolithic analysis and even compression techniques are not viable options for checking deadlock freedom.

\section{A systematic and scalable method for ensuring deadlock freedom}
\label{sec:method}

In this section, we propose a systematic approach that combines network decomposition and the application of our behavioural patterns to construct and verify some deadlock-free systems. In addition to the method itself, we propose the DFA (Deadlock-Freedom Analysis) tool to support our method's application. It is a plugin to the well-known Eclipse IDE, offering an Eclipse-like look-and-feel. It fully automates most of the application steps of our method. The only step that is not fully automated is checking pattern adherence. It involves the user selecting a pattern and providing the information needed to construct its descriptor. In Subsection 5.1 we give an overview of the DFA tool. The decomposition and patten adherence method is presented in Section 5.2, and its application, using the tool, in Subsection 5.3 (the decomposition strategy) and in Subsection 5.4 (pattern adherence). Finally,  Subsection 5.5 is dedicated to the evaluation of our method,  comparing the efficiency of the deadlock analysis of the systems developed using our approach with three other approaches.

\subsection{Deadlock Freedom Analysis tool overview}

Through this section, we use our tool to discuss and illustrate our method application. We begin by briefly describing DFA's interface and how it can be used to model a network, and then we propose our method and explain how DFA supports its application.

DFA's graphical interface is divided into four areas as depicted in Figure \ref{fig:interface}. We number the areas in this figure to facilitate referencing them. Area 1 provides the projects or networks that have been created in a given workspace. In this example, we created the networks \emph{RingBuffer} and \emph{DiningPhilosophers}. To create a project, we provide a project creation wizard in Eclipse's \emph{New} menu. Area 2 provides a view of the communication graph of the network under analysis. In this case, we selected the RingBuffer project.

Area 3 provides three panels that enables one to have an overview of the elements that have been created to construct the network, such as components, channels, etc. Area 4 has several panels that give details of the elements that have been created, and allows the user to edit them. For instance, for a selected component, it shows its alphabet, behaviour and name. In the following, we present in more detail Areas 3 and 4, their panels and the features that they offer.

\begin{figure}[t]
	\includegraphics[width=\textwidth]{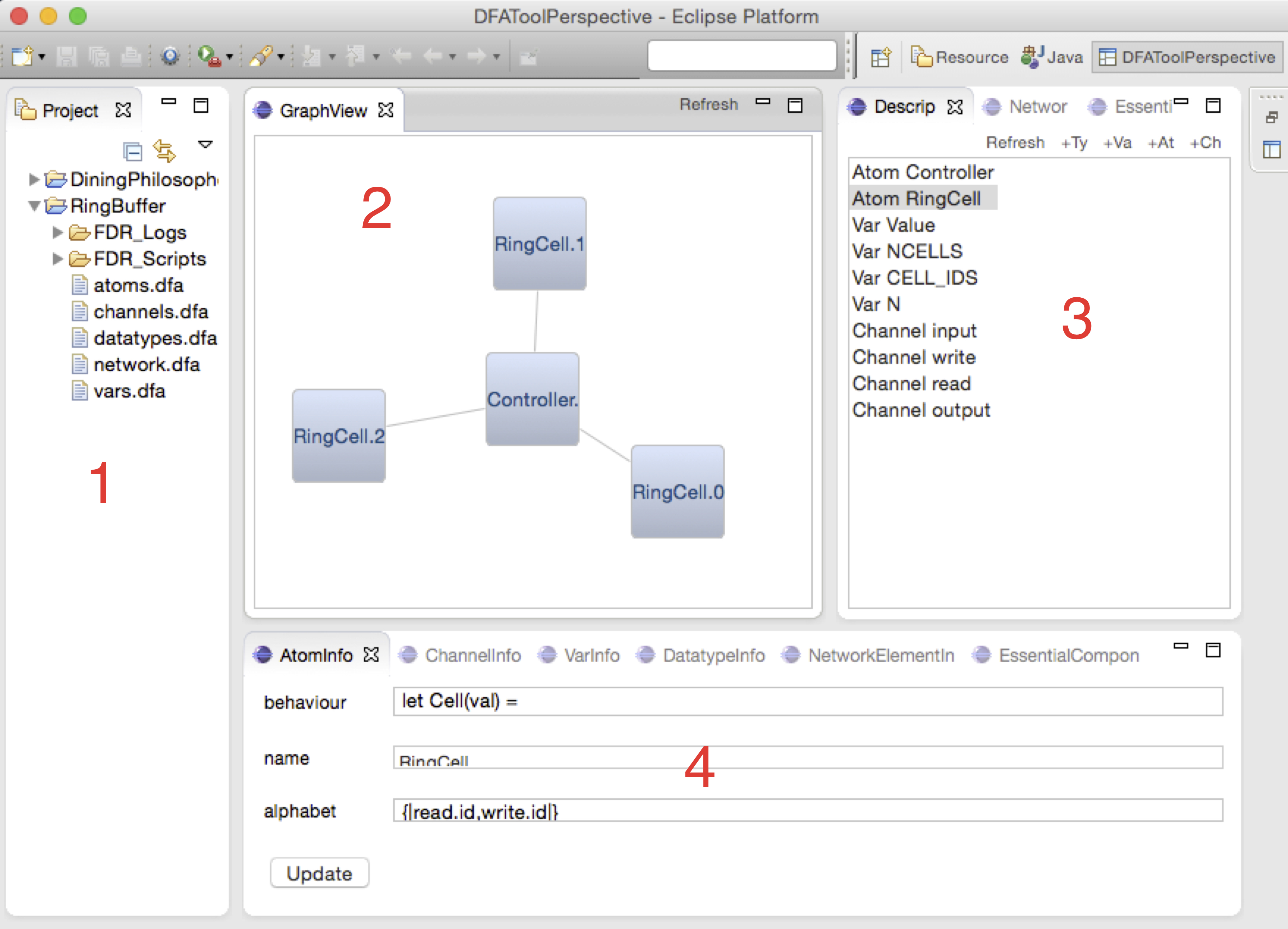}
	\caption{DFA's interface.}
	\label{fig:interface}
\end{figure}

Area 3 offers three different panels: the \emph{description-list} panel, the \emph{network-list} panel and the \emph{essential-components-list} panel. The description-list panel lists the elements that have been declared and are, as a consequence, available for the construction of the network. These elements are: atom, channel, variable, and datatype declarations. Also, in the top part of it, it has four buttons that allows the user to create new elements. An \emph{atom} (or component schema) is a parametrised component, i.e.  its alphabet and behaviour are parametrised. So, it becomes a component once the parameters are defined. A channel declaration is a declaration of a set of events, and the last two elements are self-explanatory. The network-list panel provides instantiations of atoms that define the network. The purpose of having these two separate notions for an atom and its instantiation (a component) is to facilitate the creation of networks composed of many similar components. We discuss the essential-components-list panel later.

For instance, Figure \ref{fig:area2panels} depicts the declarations and instantiations used to create our RingBuffer network. We can see that this network is composed of a single controller atom that has been instantiated with the value $0$ and three ring cell atoms that have been instantiated with values $0$, $1$ and $2$. Thus, we use a set notation to denote the parameter values (and number of components) that are to be instantiated for each atom.

\begin{figure}[t]
	\centering
	\includegraphics[width=0.35\textwidth]{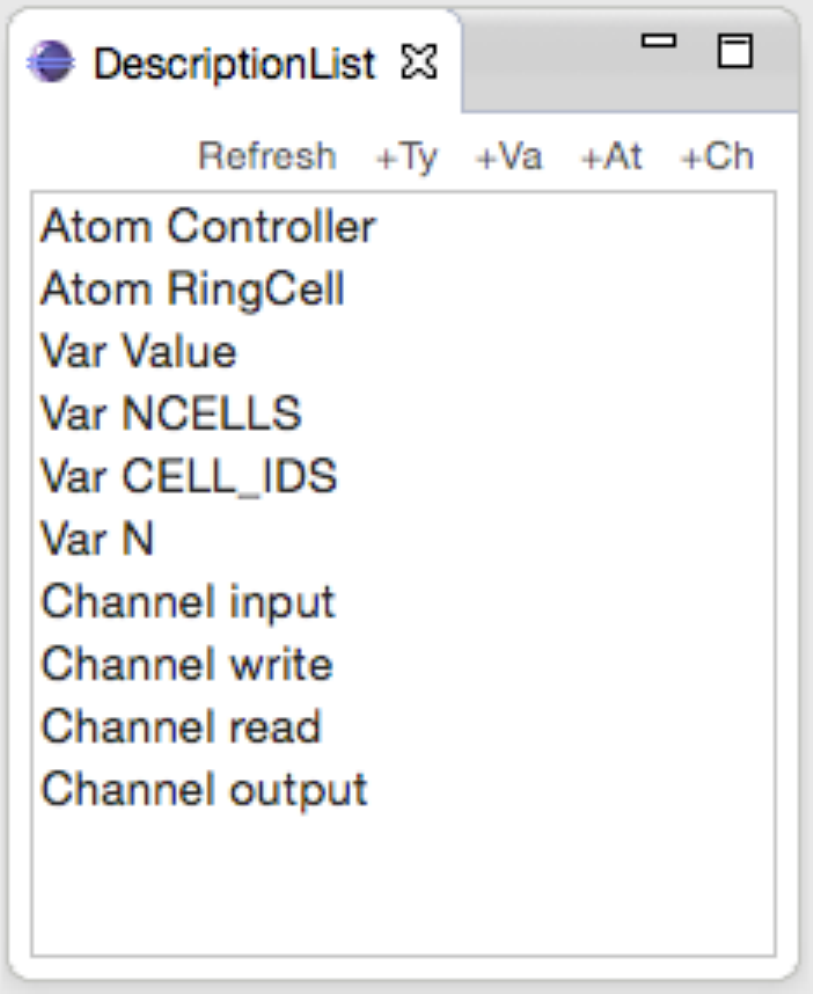}
	\hspace{1cm}
	\includegraphics[width=0.35\textwidth]{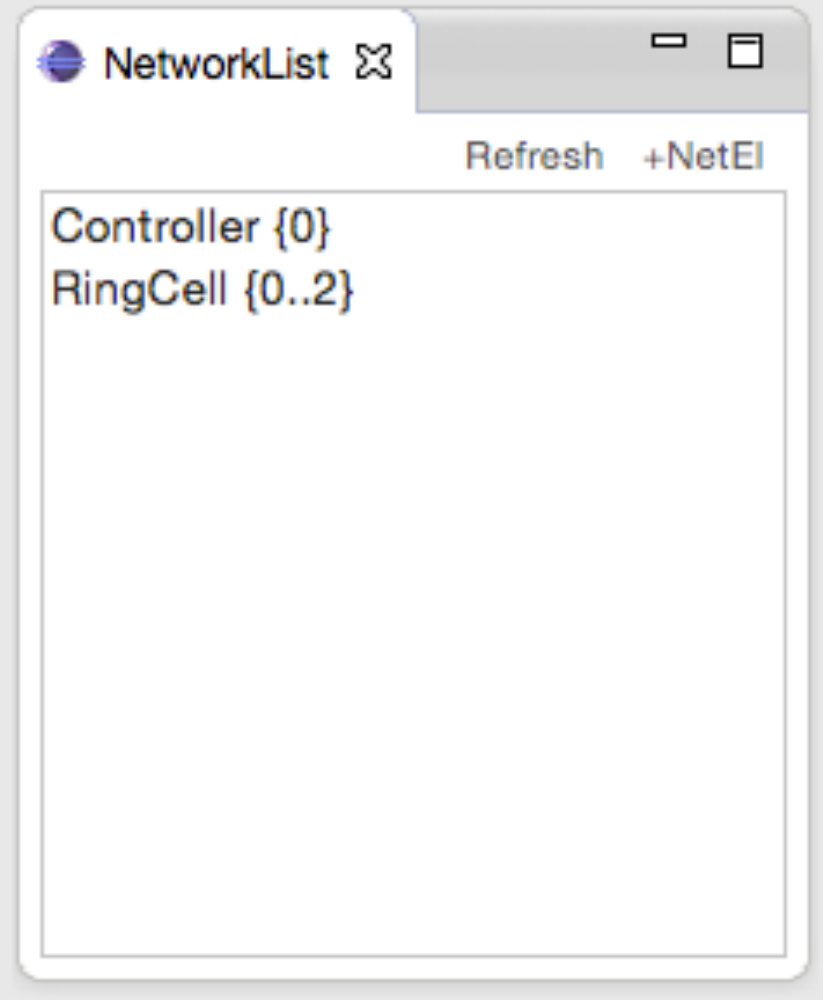}
	\caption{Descriptions-list and network-elements-list panels, respectively.}
	\label{fig:area2panels}
\end{figure}

\begin{figure}[!b]
	\centering
	\includegraphics[width=\textwidth]{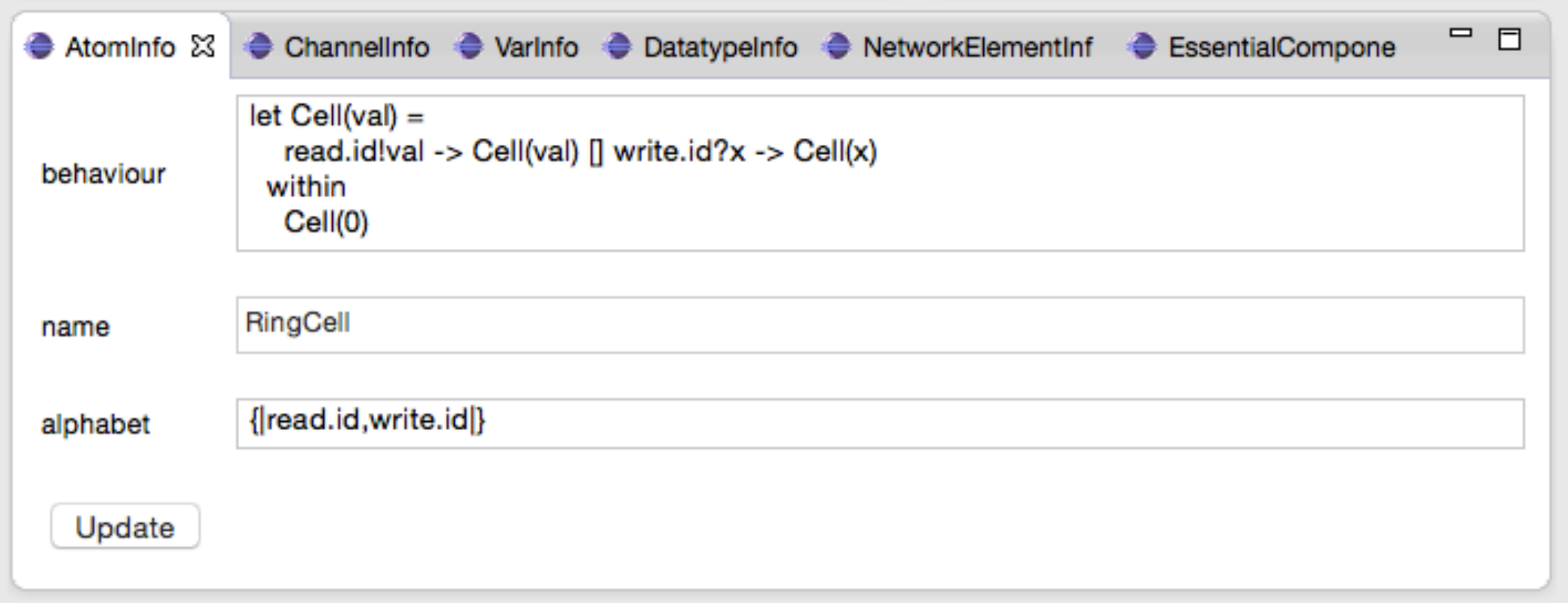}
	\caption{Atom-info panel.}
	\label{fig:atomInfo}
\end{figure}

Area 4 offers 6 panels: \emph{atom-info}, \emph{channel-info}, \emph{datatype-info}, \emph{variable-info}, \emph{network-element-info} and \emph{essential-component-info} panels. Upon selection of an atom in the descriptions-list panel, the atom-info panel presents its details. It shows its name, parametrised behaviour and parametrised alphabet. Atoms are parametrised by the implicit variable \emph{id}. This variable is what needs to be instantiated to turn an atom into a component. At the bottom of the panel, the \emph{update} button allows the user to edit the details of an atom. Panels channel-info, datatype-info, and variable-info provide similar informative and editing functionalities for the other declared elements. For instance, in Figure \ref{fig:atomInfo}, we illustrate the declaration of the RingCell component schema for the \emph{RingBuffer} network. Note the behaviour and alphabet are described using $CSP_M$ and have the implicit variable $id$. Upon selection of a network element in the network-list panel, the network-element-info provides the user with detailed information about this element and an update functionality, just like for the atom-info panel. We discuss the \emph{essential-component-info} panel later.

\subsection{The Decomposition and Pattern Adherence method}

After presenting how our tool can be used to model a network, we move on to propose and discuss our verification method. The DPA (Decomposition and Pattern Adherence) method essentially relies on two main phases: firstly, it decomposes the network, then it proves that the essential subnetworks are deadlock free. In the following, we detail all smaller steps that are necessary to carry out both of DPA's two main phases. We discuss how the steps can be implemented and estimate the complexity of this method. 

The steps of the DPA method are as follows.
	\begin{enumerate}
		\item Decompose network (identify essential subnetworks):
		\begin{enumerate}
			\item \label{s1.1} Construct communication graph;
			\item \label{s1.2} Identify disconnecting edges (bridges) in this undirected graph;
			\item \label{s1.3} Remove conflict-free disconnecting edges; and
			\item \label{s1.4} Identify resulting essential subnetworks.
		\end{enumerate}
		\item Show pattern adherence for essential subnetworks with more than one component:
		\begin{enumerate}
			\item \label{s2.1} Describe pattern descriptor for each of these subnetworks; and
			\item \label{s2.2} Check pattern adherence.
		\end{enumerate}
	\end{enumerate}

\subsubsection{Method application: decomposition strategy}
The first part of our method attempts to decompose the network under analysis. As our decomposition strategy is based on the network's topology, in Step~\ref{s1.1}, it constructs the network's communication graph. The creation of the communication graph can be carried out in time $\bigo(n^2|A|)$ where $n$ is the number of components in the network and $|A|$ over-approximates the size of individual component alphabets (say, it is the size of the largest alphabet). This approximates the time taken to create the edges of the graph. There are $\bigo(n^2)$ potential edges (pairs of components) in this graph and, for each pair of component, we can check whether their alphabets intersect, thereby giving rise to an edge in the communication graph, in $\bigo(|A|)$ steps. 

In the next step, our decomposition strategy identifies disconnecting edges. There is a linear time algorithm -- taking time $\bigo(|V|+|E|)$ where $|V|$ and $|E|$ are the sizes of the sets of nodes and edges, respectively, of the input graph -- that identifies all the bridges of an undirected graph~\cite{Tarjan74}. This algorithm can be readily applied to find disconnecting edges in a communication graph. So, it takes time $\bigo(n^2)$ to find all disconnecting edges in such a graph, given that the communication graph has $\bigo(n)$ nodes and $\bigo(n^2)$ edges.

Step~\ref{s1.3} involves finding which disconnecting edges are conflict free and removing them. So, for each pair of components corresponding to a disconnecting edge, we test them for conflict freedom using the refinement assertion in Definition~\ref{thm:conflictFreedomAssertion}. A graph with $|V|$ nodes has at most $|V|-1$ bridges. So, our communication graph has $\bigo(n)$ bridges to be tested. For the purposes of estimating DPA's complexity, we assume that components are described by labelled transition systems instead of CSP processes. This is a reasonable assumption since CSP has an operational semantics that enables this translation and most checkers internally represent components and systems in this way. We assume that $|B|$ is number of states/nodes for the largest component (i.e., transition system) of the input network. Refinement checkers work by examining the product space of specification and implementation. Our specifications, used to constrain the behaviour of network components, are small and simple processes, which should be simply normalised. So, in our complexity analysis, we factor specifications out and use the size of the implementation's state space as an estimate for the work required to check some refinement expression. Thus, if the implementation is a network with $n$ components, its has $\bigo(|B|^n)$ states and the refinement checking has to examine this many states. Note the state-space explosion is represented by the exponent $n$ in this bound. Our conflict-freedom refinement expression, however, analyses only a pair of components (placed in the \cspm{Context} process) at a time. So, checking each of our conflict-freedom refinement expressions takes $\bigo(|B|^2)$ steps\footnote{In fact, we should say that the state space of the Context process, which is the actual implementation, is proportional to $|B|^2$.}. Moreover, given that there are $\bigo(n)$ bridges, Step~\ref{s1.3} can be carried out in $\bigo(n|B|^2)$ steps.

The last step of our decomposition strategy consists of calculating the resulting essential subnetworks. This step consists of finding the graph-theoretic connected components of the graph resulting from the removal of conflict-free disconnecting edges. These connected components can be found in linear time in the size of the input communication graph using depth-first search. Therefore, similarly to Step~\ref{s1.2}, this step can be carried out in time $\bigo(n^2)$.

We use our RingBuffer network in Example~\ref{ex:buffer} to illustrate the proposed decomposition strategy. In Step~\ref{s1.1}, our strategy constructs this system's acyclic communication graph depicted in Area 2 of Figure~\ref{fig:interface}; it has a controller and three memory cells. Given its acyclic topology, Step~\ref{s1.2} finds that all its edges are disconnecting. Step~\ref{s1.3} analyses each of these edges using our conflict-freedom assertion to find out that all of them are conflict free. So, they are all removed leading to the communication graph depicted in Figure~\ref{fig:menu}. Finally, Step~\ref{s1.4} finds that each individual component is a singular essential subnetwork, i.e., an essential subnetwork with a single component.

This strategy could be manually carried out. It would, however, involve many tedious and error-prone tasks such as manually constructing and analysing a graph and manually crafting our conflict-freedom refinement assertions. Instead,
it is much more productive to carry it out in a fully automatic way by using our tool, via the \emph{Decompose} option in DFA's menu depicted in Figure~\ref{fig:menu}. It fully automates the strategy's steps using the algorithms we discuss and the FDR2 tool~\cite{FDR2}, in background, to check the conflict-freedom refinement expressions.

\begin{figure}[t]
	\centering
	\includegraphics[width=0.3\textwidth]{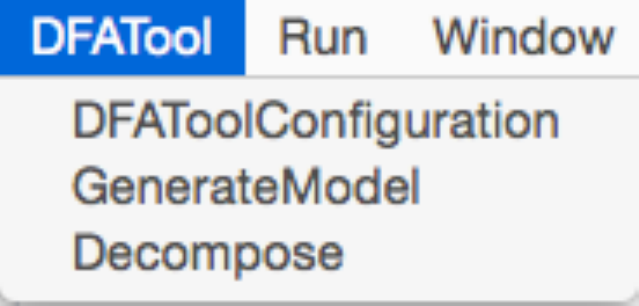}
	\hspace{1.5cm}
	\includegraphics[width=0.4\textwidth]{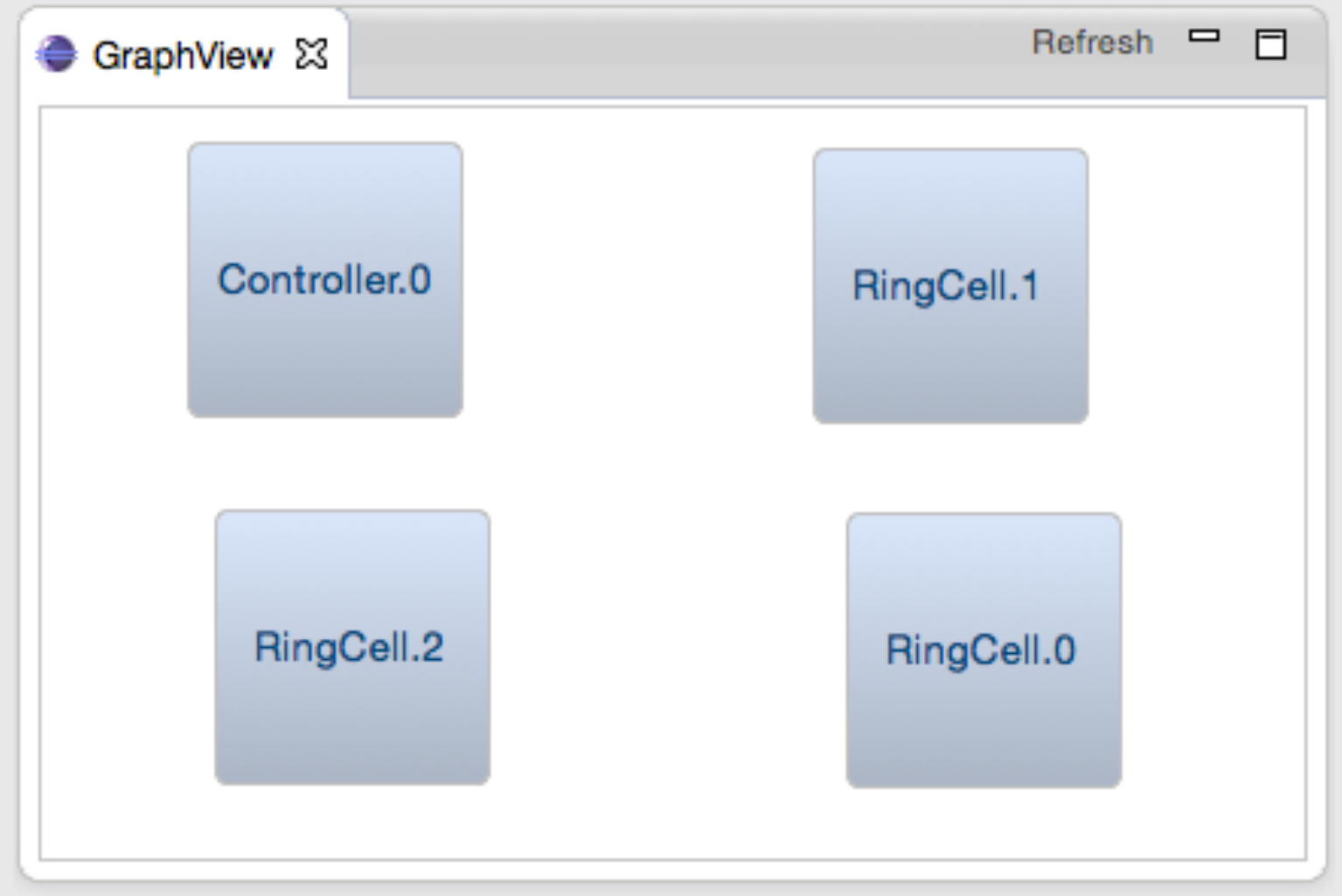}
	\caption{DFA's menu and RingBuffer's decomposed communication graph.}
	\label{fig:menu}
\end{figure} 

We make a few relevant remarks about our decomposition strategy. Firstly, this strategy alone can prove conflict-free acyclic systems deadlock free. If after decomposition all essential subnetworks are singular then the network under analysis must be deadlock free. This is exactly the case for our RingBuffer example. The buffer in this example has only three cells but our strategy can, in fact, show deadlock freedom for similar buffers with any fixed number of cells. Thirdly, based on the analysis of its complexity, this decomposition strategy seems much less computationally costly than carrying our deadlock-freedom checking for the entire network. While traditional exact deadlock-freedom checking explores the network's entire state space (taking $\bigo(|B|^n)$ steps), our strategy can be carried out in polynomial time, taking $\bigo(n^2|A||B|^2)$ steps. Therefore, it is much more scalable in proving deadlock freedom for conflict-free acyclic systems when compared to traditional exact methods. 

\subsubsection{Method application: pattern adherence}

The second part of DPA consists of proving that the essential subnetworks found by our decomposition strategy are deadlock free. As singular essential subnetworks are deadlock free by our busyness requirement, this second part is only really concerned with showing deadlock freedom for non-singular essential subnetworks and we do so via pattern adherence. Our method requires showing that each of these non-singular essential subnetworks adhere to one of our patterns. So, for each of these networks, the user of our method has to choose which pattern it adheres to and provide the appropriate pattern descriptor. Given a pattern descriptor, one can simply test adherence by validating the structural and behavioural constraints. The behavioural constraints can be validated using the refinement expressions we propose, whereas structural restrictions can be tackled by simple iterative algorithms. In the following, we discuss and illustrate this part of our method and the tool support we provide using the DiningPhilosopher network in Example~\ref{ex:philosophers} and the resource-allocation pattern.

Given its ring-like topology, as depicted in Figure~\ref{fig:ecInfo0}, the DiningPhisolophers network has no disconnecting edges. So, the application of the decomposition strategy to it results in the original network being the single essential subnetwork found. In our tool, the decomposition strategy updates the \emph{essential-components} panel (located in Area 3 of Figure~\ref{fig:interface}) to show the non-singular essential subnetworks found. In our example, Figure~\ref{fig:ecInfo0} shows that DFA finds this single essential subnetwork and names it \emph{EC0}.

\begin{figure}[!t]
	\centering
	\includegraphics[width=0.58\textwidth]{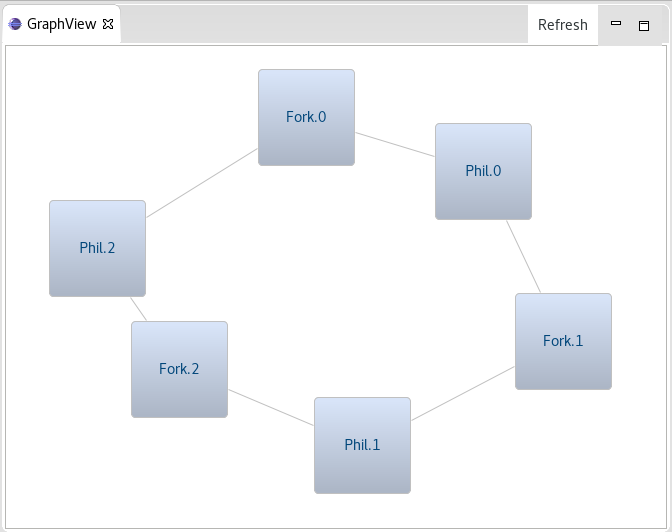}
	\hspace{0.2cm}
	\includegraphics[width=0.38\textwidth]{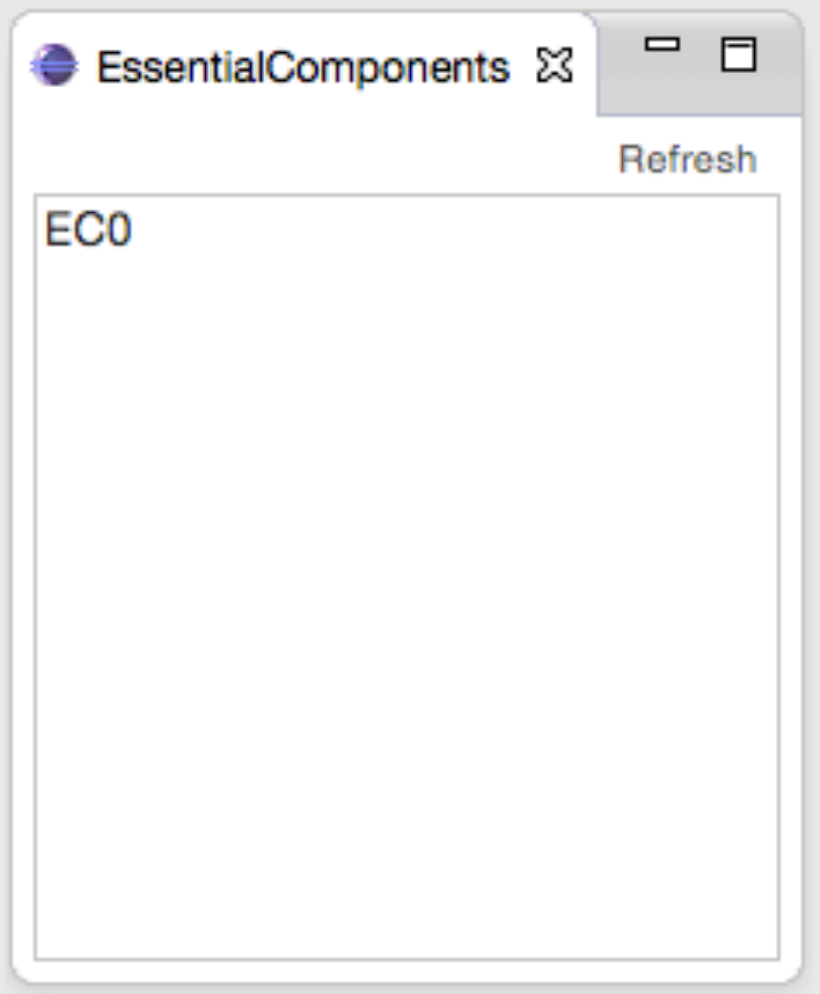}
	\caption{DiningPhilosophers' communication graph and essential-components panel.}
	\label{fig:ecInfo0}
\end{figure}

By our method's definition, we are then left with proving this essential subnetwork, i.e., the entire original network, adheres to some pattern; we show it adheres to the resource-allocation pattern. Firstly, we identify the resource allocation descriptor for this network. Instead of describing the descriptor in terms of the structure $\mathcal{C}$ as per Section~\ref{sec:patterns}, we directly describe sets users and resources, and functions users, resources, acquire and release. Also, we represent our $>_{RA}$ order by a sequence of resources $\Seq{r_1,\ldots,r_n}$, where $r_i >_{RA} r_j$ if and only if $i > j$. This alternative description is the one used by our tool. We believe that, although less concise, this alternative description is more user-friendly and more suited to users that are not formal methods enthusiasts.

\begin{defn} \label{def:ra_meta} Resource allocation descriptor for DiningPhilosophers. We use \verb|Phil.i| to identify a philosopher component, \verb|Fork.i| a fork one, and $N = 3$ respresents the number of philosophers/forks in the network.
\begin{itemize}
	\item $\textit{User} =$ \verb"{Phil.i | i <- {0..N-1}}"
	\item $Resources = $ \verb"{Fork.i | i <- {0..N-1}}"
	\item $users(id) = $ \verb"{Phil.id,Phil.((id-1)%N)}"
	\item $resources(id) = $ \verb|if id == N-1 then <Fork.0,Fork.id>| \\ \hspace*{2.5cm} \verb|else <Fork.id,Fork.((id+1)%N)>|
	\item $acquire(idU,idR) = $ \verb|pickup.idU.idR|
	\item $release(idU,idR) = $ \verb|putdown.idU.idR|
	\item $>_{RA} =$ \verb"<Fork.i | i <- <0..N>>"
\end{itemize}
\end{defn}

We can test whether the provided descriptor satisfies the pattern's structural restrictions using simple iterative algorithms. For instance, the condition \emph{partitioned} can be checked in time $\bigo(n)$ by a simple algorithm that carries out the required operations and comparisons on the two sets: \emph{users} and \emph{resources}. Similarly, we can check \emph{controlled\_alpha\_users} in time $\bigo(n^2|A|^2)$ since $\textit{Voc}$'s size is bound by $\bigo(n|A|)$; we can iterate over $\textit{Voc}$ at most $|A|$ times to find the intersection set $\textit{Voc} \inter A$, and there are $n$ such calculations to be carried out. 

This descriptor also gives the information that we need to craft the appropriate refinement expressions to test whether the behavioural constraints the pattern enforces are met. In the case of our DiningPhilosopher network, it leads to $N$ assertions for philosophers and $N$ for forks. Since each assertion checks a component individually they can be carried in time $\bigo(|B|)$, and checking all of them takes time $\bigo(n|B|)$. 

\begin{figure}[h]
	\centering
	\includegraphics[width=\textwidth]{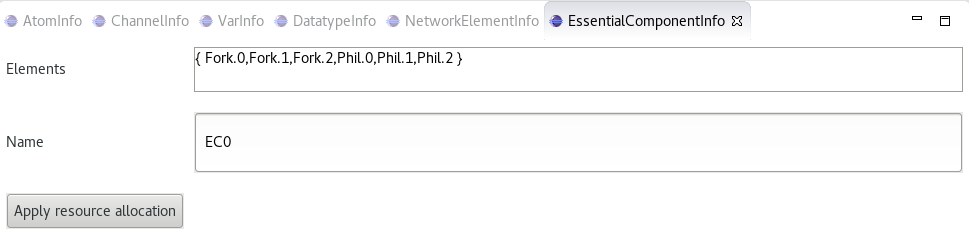}
	\caption{Essential-component-info panel for ECO subnetwork.}
		\label{fig:ecInfo}
\end{figure}

Our tool supports this step as follows. By selecting an essential subnetwork in the essential-components panel, the \emph{essential-component-info} panel (located in Area 4 of Figure~\ref{fig:interface}) is updated to show the components in this essential subnetwork. For instance, Figure~\ref{fig:ecInfo} presents the information for our example's EC0 subnetwork. This panel also allows the user to apply the resource allocation by clicking on the ``Apply resource allocation" button. Then, the user has to input the pattern descriptor via a dialog box as depicted in Figure~\ref{fig:descriptor}. The boxes should be filled as per Definition~\ref{def:ra_meta}\footnote{Our tool requires the two sets \emph{users} and \emph{resources} to be written without the variable $N$ so intervals are \texttt{0..2} instead of \texttt{0..N-1}. Also, we adopt the convention that $<_{RA}$ is the natural order on Fork's identifiers.}. At the moment, our prototype only supports the application of the resource allocation pattern. Other patterns can be similarly implemented using the same core idea. Given this descriptor, our tool can show that this subnetwork adheres to the resource allocation pattern, and so the network is deadlock free. This network has only 3 philosophers and forks but our method can, similarly, tackle this example for any fixed number of philosophers and forks.

\begin{figure}[b]
	\centering
	\includegraphics[width=\textwidth]{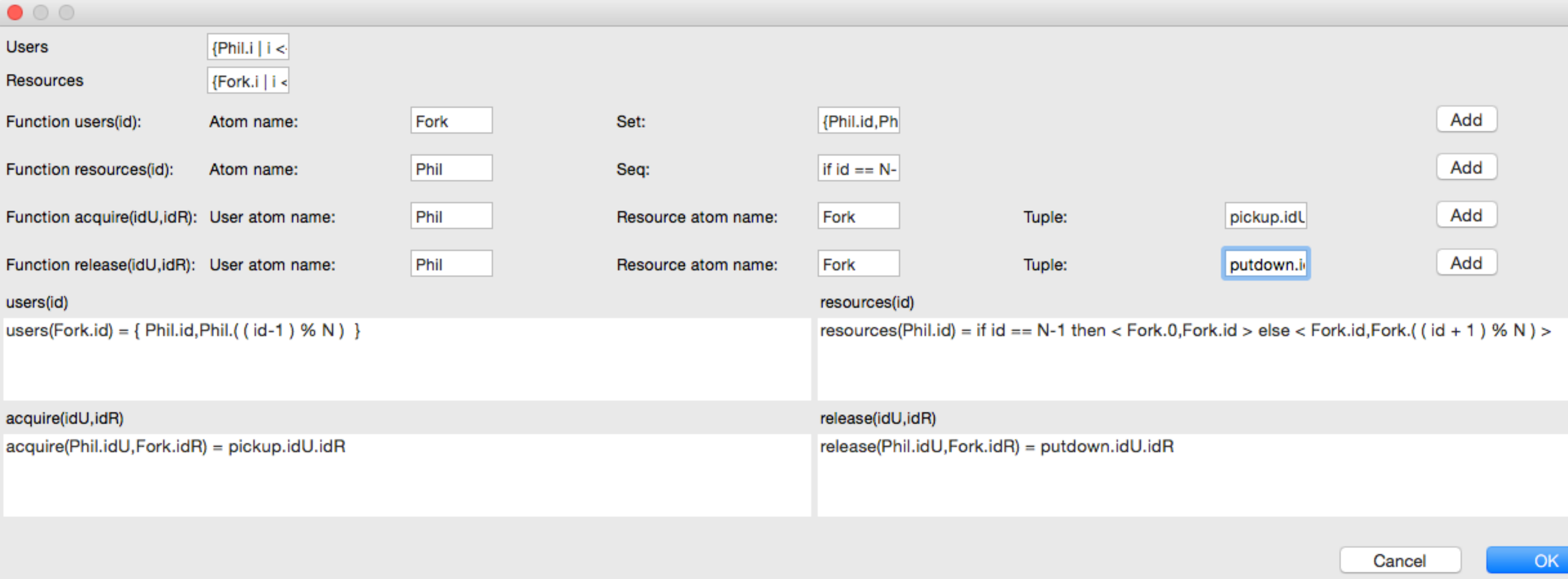}
	\caption{Resource-allocation descriptor dialog box.}
	\label{fig:descriptor}
\end{figure}

We point out that in terms of efficiency, pattern adherence checking should be substantially faster than monolithically checking deadlock freedom for the corresponding network. For the cases when the state space of a system increases exponentially with the number of components, our approach will very much outperform monolithic approaches. While monolithic deadlock checking has to explore an exponentially large state space in general, pattern-adherence verification only examines one component at a time, for behavioural conditions, and the structural conditions can be polynomially checked in the size of the structure of the process (i.e., size of alphabets and number of nodes), which tends to be much smaller than the behavioural part. On the other hand, there are concurrent systems for which state space only grows polynomially. They are not common but they exist. So, in general, as the state-explosion problem affects most (interesting, worth-verifying) concurrent systems, our approach should normally outperform monolithic ones.

\subsection{Method evaluation}
\label{sec:evaluation}
In this section, we empirically evaluate our method. Our evaluation only takes into account the verification of behavioural constraints. So, the verification times that we present for DPA disregard the examination of structural restrictions. Checking the behavioural aspect of our method should be much more demanding than checking its structural counterpart, given the static nature of a network's structure and the simplicity of structural conditions. Thus, the time to verify our method's behavioural conditions should approximate the time that would take checking structural restrictions.


We compare DPA against three other approaches that can prove deadlock freedom: the \emph{SDD} framework implemented in the Deadlock Checker tool~\cite{DC}, FDR2's built-in deadlock-freedom assertion (FDR2) and its combination with compression techniques (FDR2c)\footnote{FDR is currently in its fourth version (FDR4). Version 3 was a complete rewrite of FDR2 which largely improved it. This version (and subsequent ones), however, does not implement the stable-revivals model, which is an essential part of our method's conflict-freedom analysis.}. SDD is an incomplete framework that works by constructing the system's dependency digraph and checking it for cycles. A live system/network that does not exhibit such a cycle must be deadlock free~\cite{Martin96}. FDR2 and FDR2c are complete methods that explicitly explore the system's state space. While FDR2 simply explore this space, FDR2c relies on some user-provided hierarchical compression strategy to attempt to reduce the size of the system's original state space. We point out that while incomplete methods can only show deadlock freedom for some deadlock-free systems, complete ones do so for them all. This incompleteness is the price paid to achieve efficiency.

We use our two running examples in this comparison, i.e., the ring buffer and the asymmetric dining philosophers examples, and we also check the transport layer of the leadership-election system presented in~\cite{Antonino14a}. This last example models a protocol used by B\&O's\footnote{\url{http://www.bang-olufsen.com/}} Audio and Video (AV) systems to elect a leader that coordinates the interaction between components. In this system, nodes exchange messages containing their priority value, which represents their eagerness to become the leader, so they coordinate and agree on which node becomes the leader. The transport layer is composed of the nodes themselves and bus cells implementing the asynchronous means of communication through which they exchange messages. This pattern adheres to the async-dynamic pattern where nodes are participants and bus cells are transport entities. A detailed description of this system and its adherence to this pattern can be found in the aforementioned work. Our evaluation was conducted on a dedicated machine with Intel i7-7500U CPU @ 2.70GHz, 16GB of RAM, and running Fedora 25, and the scripts used can be found in~\cite{DPAEXP}.

\begin{table}[!t]
	\centering
		\begin{tabular}{|c|c|c|c|c|c|}
			\cline{2-6}
		     \multicolumn{1}{c|}{}& n     & DPA       & SDD       & FDR2      & FDR2c     \\ 
			\hline 
			\multirow{5}{*}{Ring buffer} & 3     & 0.01      & 0.25      & 0.03      & 0.06      \\ 
			\cline{2-6} 
			& 5     & 0.32      & 0.28      & 0.27      & 0.03      \\ 
			\cline{2-6} 
			& 10    & 1.54      & 1.83      & 371.94    & 0.11      \\ 
			\cline{2-6}
			& 20    & 11.65     & 48.28     & -         & 0.42      \\ 
			\cline{2-6} 
			& 30    & 53.43     & -         & -         & 0.92      \\
			\hline
			\multirow{5}{*}{Dining philosophers} & 5     & 0.03      & 0.18      & 0.06      & 0.03      \\ 
			\cline{2-6} 
			&10    & 0.11      & 0.18      & 448.41    & 0.06      \\ 
			\cline{2-6}  
			&50    & 3.82      & 0.18      & -         & 0.82      \\ 
			\cline{2-6}  
			&100   & 28.02     & 0.23      & -         & 7.28      \\ 
			\cline{2-6}  
			&200   & 214.13    & 0.38      & -    	  & -         \\ 
			\hline 
			\multirow{5}{*}{Leadership Election} & 3     & 0.30      & -         & 2039.40   & +         \\ 
			\cline{2-6} 
			&5     & 0.89      & -         & -         & +         \\ 
			\cline{2-6} 
			&10    & 7.76      & -         & -         & +         \\ 
			\cline{2-6} 
			&20    & 71.54     & -         & -         & +         \\ 
			\cline{2-6} 
			&30    & 383.30    & -         & -         & +         \\ 
			\hline
		\end{tabular}
		\caption{Results of evaluation. $n$ is a parameter used to configure the size of the systems tested. We measure in seconds the time taken to check deadlock freedom for each system. - means that the method could not prove deadlock freedom for the system: either the (incomplete) method is unable to prove so, or it took longer than 1 hour, or an error, such as running out of memory, occurred. + means that no efficient compression technique could be found.}
		\label{tab:results}
	\end{table}
	
The results of our evaluation are presented in Table~\ref{tab:results}. For DPA, we show the results of showing conflict freedom for the disconnecting edges of examples RB, behavioural adherence of examples DP to the resource allocation pattern, and behaviour adherence of examples LE to the async-dynamic pattern.

Unsurprisingly, these results suggest that incomplete methods are fairly scalable; both DPA and SDD can handle these examples quite efficiently, albeit SDD cannot tackle the leadership-election examples we analyse. FDR2's built-in assertion quickly becomes unable to handle systems with the growth of $n$. This demonstrates the state-explosion problem in practice. Its combination with compression techniques, however, is fairly effective in handling the ring buffer and dining philosophers examples. We point out that the user has to find a good compression strategy to make this approach effective and find such a strategy is not usually a simple task. For instance, for the leadership-election system, we were unable to find a good compression strategy. These results also suggest that our method can tackle systems that cannot be handled by traditional incomplete methods such as SDD and that the sort of local analysis that our method employs might be the only alternative in handling complex systems such as the ones modelled in the leadership-election examples.

We reinforce that unlike the other approaches, DPA provides not only a method to check that a system is deadlock free but a guideline to construct deadlock-free systems. In fact, our formal analysis of B\&O's protocol identified several issues that were addressed by modifications to the real C++ implementation, which were guided by our async-dynamic pattern. This attests the real and practical impact that our method can make.

\subsection{Final considerations}

The main driving force behind our method's efficiency is \emph{local analysis}. Instead of explicitly examining the global behaviour of the system as traditional approaches do, we only analyse small parts of the network at a time. Our method analyses the behaviour of pairs of components when we analyse disconnecting edges for conflict freedom and of individual components when checking for pattern adherence. Our method can be seen as a systematisation of local reasoning to ensure deadlock freedom. As we show later, for some complex networks, local reasoning might be the only practical alternative for guaranteeing deadlock freedom.

Our work was inspired by~\cite{Roscoe87} and~\cite{Martin96}. They proposed the ideas behind decomposition and pattern adherence, and we refined, combined and extended them into a practical framework. The soundness of our method follows straightforwardly from Theorems~\ref{thm:decomposition},~\ref{thm:conflictFreedomAssertion},~\ref{thm:ra},~\ref{thm:cs} and~\ref{thm:ad}. Our method can show deadlock freedom for acyclic systems that are conflict free and for cyclic systems for which essential subnetworks adhere to one of our pattern. Unlike traditional techniques that propose \emph{a posteriori} verification, our method proposes an approach that can be used as a guide to design deadlock-free systems. One can create a deadlock free acyclic network by composing components in a way that they are conflict free; conflict free captures the natural idea that for an effective communication protocols must be conjugate/symmetric, if a component requests some action its communication partner must provide it. Also, one can design cyclic networks by ensuring that they conform to one of our patterns. Note that this perspective also leads to a way to combine different patterns into a single complex system. If we have (sub)networks that adhere to different patterns, we can link them with conflict-free disconnecting edges and the resulting network is also deadlock free. So, our method also allows for this sort of combination of patterns.

\section{Related Work}
\label{sec:related}

In this section, we discuss some alternative incomplete approaches to ensure deadlock freedom. Broadly speaking, we can split such approaches into non-constructive and constructive ones. Constructive approaches explicitly provide guidelines on how to construct deadlock-free networks, whereas non-constructive approaches do not. So, while non-constructive approaches normally propose some \emph{a posteriori} verification technique, constructive ones provide some systematic technique to avoid deadlocks. We begin by analysing constructive approaches, and then discuss non-constructive ones.

\subsection{Constructive approaches}

Roscoe and Brookes developed a theory, which is used in this work, for analysing deadlock freedom for networks of CSP processes~\cite{Brookes91}. They identified a cycle of ungranted requests as a necessary condition for a deadlock. Roscoe and Dathi contributed by developing a (local) proof method for deadlock freedom~\cite{Roscoe87}. They have built a method to prove deadlock freedom based on variants, similar to the ones used to prove loop termination. In their work, they also analyse some patterns that arise in deadlock free systems. They use the proposed proof rule to establish deadlock freedom for some classes of networks. 

Following these initial works, Martin defined and formalised some design rules to guarantee deadlock freedom by avoiding cycles of ungranted requests~\cite{Martin96}. These design rules are similar to our patterns in the sense that they describe some constraints to be followed while designing a network so as to avoid deadlocks. Nevertheless, they describe behavioural constraints as semantic properties that processes in the network must have, and no automatic way of checking design rule adherence is suggested.

In~\cite{Isobe05}, the authors propose an encoding of the network model and of a proof rule from \cite{Roscoe87} in a theorem prover. Even though this encoding provides mechanical support for deadlock analysis and allows one to reason locally, it does not resolve the problems that motivated this work, which is to insulate the user as much as possible from the details of the formalisation. For instance, in order to carry out the proof using this approach one has to understand the stable-failures semantic model, has to directly interact with the theorem prover, and has to provide some mathematical structures that are not evident, such as a partial order that breaks possible cycles of ungranted requests. On the other hand, our work could benefit from this encoding to mechanise the formalisation of our patterns using a theorem prover.

In \cite{Lambertz13}, a method that proves deadlock freedom for message-passing component-based systems is proposed. This method only deals with live networks that respect some topological restrictions. It presents a necessary condition for a deadlock based on the analysis of wait-for dependencies for pairs of components. So, this condition can be checked in polynomial time, which also implies that this method is immune to the state space explosion problem. No automated strategy, however, is proposed to verify that a given network respects these restrictions.


BRICK~\cite{Ramos09,Ramos11} is an alternative approach for designing asynchronous deadlock-free systems. This approach represents systems as contracts and proposes rules for composition of systems that ensure deadlock freedom. 
BRICK is systematic and rely on refinement expressions to discharge the side conditions imposed by composition rules. A BRICK user can create basic contracts from scratch, and then design design deadlock-free systems, in a step-by-step fashion, guided by the proposed rules. This approach, however, is not fully compositional.
One of its composition rules, the $reflexive$ rule, imposes a restriction on the overall behaviour of the resulting composition, rather than on its components, like the other rules require. As this composition is a parallel combination of components, this verification can quickly become unfeasible. This issue is rather significant given that cyclic-topology networks can only be created using this rule. In \cite{Oliveira2016}, we adapted our pattern based approach to BRICK to make this rule compositional. Recently, a tool to support the original BRICK framework without pattern adherence has been proposed~\cite{Pereira17}.

In \cite{Lynch81}, the author studies networks that conform to the resource-allocation pattern. The author acknowledges that to make such a network deadlock-free, one has to impose a strict order on the way resources are acquired. She goes, then, into studying how to choose a good ordering of resources in the sense that it minimises the time users need to wait to acquire resources. She proposes a few ``good" orderings, an algorithm to implement this well-behaved acquisition of resources, and analysis of a few networks using different orderings. The problem studied in that paper is much narrower than the one we tackle here. There, to some extent, the author is refining the ordering of resources used in the resource-allocation pattern ($>_{RA}$) and finding orderings and algorithms that will maximise the work of users, by minimising the time they wait to acquire resources. So, one could potentially create a refined version of our pattern that would require component to conform to these ``good" orderings instead of a general ordering.       

As for these approaches, our method is also constructive and provide a clear and systematic guideline on how to design deadlock-free systems. Unlike these approaches, however, we do provide fully automatic procedures that can be readily implemented and used to show that a system was constructed respecting our method. Our method can be seen as a lightweight and synchronous version of BRICK that is fully local/compositional and automated. We are not aware of any other approach that automates communication pattern adherence, as we do in the DPA method and in the DFA tool. 

\subsection{Non-constructive approaches}

In addition to design rules, Martin developed three frameworks (SDD, CSDD, and FSDD) and a tool with the specific purpose of deadlock verification, the Deadlock checker~\cite{DC,Martin97}. Broadly speaking, this tool reduces the problem of deadlock checking to the quest of cycles of ungranted requests in live networks. So, it can verify deadlock freedom for some networks in a very efficient way. In fact, this method constructs a digraph, in polynomial time in the size of the input system, using local analyses of the network. Furthermore, cycle finding can be conducted in polynomial time in the size of this digraph. We point out that our method and Martin's have incomparable accuracy: some networks that can be proved deadlock free by the Deadlock Checker do not obey any of the patterns, and some networks that obey the Async-dynamic pattern cannot be proved deadlock free by the Deadlock Checker. For instance, the leadership-election system we evaluated have cycles of dependencies between participants and transport entities, rendering SDD, FSDD and CSDD unable to prove it deadlock free. 

Similarly to Martin's approaches, the techniques in~\cite{Attie05,Attie13,Attie18} rely on a graph-like structure that depicts wait-for dependencies between component states. These works prove deadlock freedom by showing that a necessary condition for the existence of this graph-like structure is not met by the system under analysis. While~\cite{Attie05} uses this framework to analyse shared-memory concurrent programs, \cite{Attie13,Attie18} extend this approach to a more general setting. In the context of shared-memory concurrent programs, this condition is shown to be checked by the analysis of pairs of components, while in the setting of~\cite{Attie13,Attie18} it is unclear the complexity for establishing this condition.

In \cite{Bensalem11,Bensalem16}, the authors propose a compositional verification strategy together with a tool, which checks deadlock freedom, based on component and global invariants; these global invariants, which are called interaction invariants, express global synchronisation requirements between atomic components. The D-Finder tool iteratively tries to find a system invariant, combining component and interaction invariants, that can ensure deadlock freedom. Although powerful, these strategies can suffer from combinatorial explosion in calculating interaction invariants.

In \cite{Antonino16a,Antonino16b,Antonino17a,Antonino17b}, the first author has proposed a number of techniques that find sophisticated invariants and use them to prove deadlock freedom. It uses local analysis to find local and global invariants that are combined to over-approximate the state space of a system. Although these approaches can be hindered by combinatorial explosion, they tend to generally be much more efficient than complete methods. These frameworks, however, cannot prove deadlock freedom for some systems that our method can. For instance, our leadership-election example cannot be proved deadlock free by these techniques; the invariants captured by these methods are not strong enough to show deadlock freedom for this system.

These non-constructive approaches and our method should come close in terms of scalability. They differ, however, in terms of the methodology employed. While these approaches try to establish deadlock freedom for a constructed system, our method provides a design guideline to help the user build a deadlock-free network. We point out that while these other approaches are fully automatic, our method is only semi-automatic in the sense that the user might be required to provide a pattern descriptor so pattern adherence can be checked. Nevertheless, one should note that such information should be trivially known to the user if they use our method, and in particular our patterns, to design the network under analysis.

\section{Conclusion}
\label{sec:conclusion}

In this work, we propose a method that combines both a decomposition strategy and behavioural-pattern adherence to prove deadlock freedom. This method can be a very useful design tool as it provides both a systematic guideline to construct deadlock-free systems and procedures to ensure that the guideline has been correctly followed. Our use of refinement expressions to impose behavioural constraints improves previous pattern formalisations in two ways. Firstly, the refinement expressions give a practical representation of the behavioural restrictions imposed by a given pattern. That means that, instead of describing semantic properties of the process, we have a CSP process describing what is expected from the behaviour of a component. Secondly, it allows automatic checking of these constraints by the use of a refinement checker.

Our method can be seen as a systematisation of local analysis to prove deadlock freedom. Local analysis is a core factor in making our method efficient. Many frameworks using local analysis have been proposed. Some of them propose a posteriori verification and give no indication of how to avoid deadlocks, whereas others provide guidelines to avoid deadlock but no automatic verification to ensure that the guidelines were correctly followed. Our method provides both a guideline and verification procedures to ensure it was properly followed. Its use in the design of a practical protocol for B\&O also demonstrates our method's practical relevance and impact. Moreover, we also provide a tool that support and automates the application of our method.  Finally, our evaluation suggests that, for some examples, our method might be the only practical and capable option to prove deadlock freedom. So, it can tackle a class of systems that cannot be handled by available incomplete approaches.

In order to improve this framework, our pattern catalogue could be augmented. Some patterns described in prior works have not been formalised in our framework yet. To make our framework more general, we plan to add those patterns to it. Moreover, the tool developed is a proof of concept that, so far, has only a single pattern available. So, a natural extension would be to add the two missing patterns to it. Finally, we plan to promote this framework to a general modelling language level, such as SysML. This involves defining a suitable component model for SysML, and adapting the proposed DPA method. It is also required a front-end tool to translate from SysML to CSP, running FDR in background, and supporting traceability between the SysML and the CSP models. This should hide the formal methods part of our method, making it more accessible for industry partners.

\section*{References}

\bibliography{references}

\end{document}